\documentclass[11pt]{article}
\usepackage[margin=1in]{geometry}
\usepackage{amsmath,amssymb,amsthm}
\usepackage{xcolor}
\usepackage{booktabs}
\usepackage{url}

\usepackage{algorithm}
\usepackage[noend]{algpseudocode}
\usepackage{graphicx} 
\usepackage{paralist}
\usepackage{enumitem}
\usepackage{caption}
\usepackage{subcaption}
\usepackage{tikz}
\usetikzlibrary{shapes,arrows}
\usepackage{xspace}
\usepackage{ifthen}
\usepackage{mdframed}
\usepackage{thmtools}
\usepackage{thm-restate}
\usepackage[most]{tcolorbox}
\usepackage{wrapfig}
\usepackage{soul}
\usepackage{hyperref}

\newtheorem{theorem}{Theorem}[section]
\newtheorem{lemma}[theorem]{Lemma}
\newtheorem{corollary}[theorem]{Corollary}
\theoremstyle{definition}
\newtheorem{definition}[theorem]{Definition}
\theoremstyle{plain}

\newtcolorbox{problem}[1]{
    colback=white,
    colframe=black,
    sharp corners,
    boxrule=0.8pt,
    fonttitle=\bfseries,
    coltitle=black,
    attach title to upper, 
    after title={\smallskip\\},
    title={#1}
}

\begin{document}

\newcommand{\ct}[1]{c(#1)}

\newcommand{\degr}[1]{d(#1)}
\newcommand{\edeg}[1]{d'(#1)}
\newcommand{\Nparents}[1]{d_a'(#1)}
\newcommand{\Nchildren}[1]{d_d'(#1)}
\newcommand{\Nintersect}[1]{d_i'(#1)}
\newcommand{\degen}{\alpha}
\newcommand{\edegen}{\kappa}
\newcommand{\tdegen}{\overline{\kappa}}
\newcommand{\trimmed}[2]{#1\langle #2 \rangle}
\newcommand{\degeneracy}{T-degeneracy}
\newcommand{\hyperDegen}{hyperedge degeneracy}
\newcommand{\HyperDegen}{Hyperedge degeneracy}
\newcommand{\ws}{WS^*}
\newcommand{\multihyper}{multi-trimmed subhypergraph}
\newcommand{\Multihyper}{Multi-trimmed subhypergraph}
\newcommand{\dah}{DAH}

\newcommand{\din}[1]{d^-_{#1}}
\newcommand{\dout}[1]{d^+_{#1}}
\newcommand{\Nout}[1]{N^+_{#1}}
\newcommand{\parents}[1]{A(#1)}
\newcommand{\children}[1]{D(#1)}
\newcommand{\sink}[1]{\ell(#1)}
\newcommand{\source}[1]{f(#1)}
\newcommand{\localIndegrees}{li}
\newcommand{\compCounts}{\gamma}
\newcommand{\cT}{\mathcal{T}}
\newcommand{\rank}{r}
\newcommand{\type}{\tau}
\newcommand{\inputsize}{h}

\newcommand{\todo}[1]{{\bf undefined comment} #1}

\newcommand{\Hom}[2]{\mathrm{Hom}_{#2}(#1)}
\newcommand{\Homr}[2]{\mathrm{rHom}_{#2}(#1)}
\newcommand{\Sub}[2]{\mathrm{Sub}_{#2}(#1)}
\newcommand{\IndSub}[2]{\mathrm{IndSub}_{#2}(#1)}

\newcommand{\ColHom}[2]{\mathrm{\text{Col-}Hom}_{#2}(#1)}

\newtheorem{claim}[theorem]{Claim}


\newcommand{\cA}{{\cal A}}
\newcommand{\cB}{\mathcal{B}}
\newcommand{\cC}{{\cal C}}
\newcommand{\cD}{\mathcal{D}}
\newcommand{\cE}{{\cal E}}
\newcommand{\cF}{\mathcal{F}}
\newcommand{\cG}{\mathcal{G}}
\newcommand{\cH}{{\cal H}}
\newcommand{\cI}{{\cal I}}
\newcommand{\cJ}{{\cal J}}
\newcommand{\cL}{{\cal L}}
\newcommand{\cM}{{\cal M}}
\newcommand{\cO}{\mathcal{O}}
\newcommand{\cP}{\mathcal{P}}
\newcommand{\cQ}{\mathcal{Q}}
\newcommand{\cR}{{\cal R}}
\newcommand{\cS}{\mathcal{S}}
\newcommand{\cU}{{\cal U}}
\newcommand{\cV}{{\cal V}}
\newcommand{\cW}{{\cal W}}
\newcommand{\cX}{{\cal X}}

\newcommand{\R}{\mathbb R}
\newcommand{\F}{\mathbb F}
\newcommand{\Z}{{\mathbb Z}}
\newcommand{\eps}{\varepsilon}
\newcommand{\lam}{\lambda}
\newcommand{\sgn}{\mathrm{sgn}}
\newcommand{\poly}{\mathrm{poly}}
\newcommand{\polylog}{\mathrm{polylog}}
\newcommand{\littlesum}{\mathop{{\textstyle \sum}}}
\newcommand{\half}{1/2}
\newcommand{\la}{\langle}
\newcommand{\ra}{\rangle}
\newcommand{\wh}{\widehat}
\newcommand{\wt}{\widetilde}
\newcommand{\calE}{{\cal E}}
\newcommand{\calL}{{\cal L}}
\newcommand{\calF}{{\cal F}}
\newcommand{\calW}{{\cal W}}
\newcommand{\calH}{{\cal H}}
\newcommand{\calN}{{\cal N}}
\newcommand{\calO}{{\cal O}}
\newcommand{\calP}{{\cal P}}
\newcommand{\calV}{{\cal V}}
\newcommand{\calS}{{\cal S}}
\newcommand{\calT}{{\cal T}}
\newcommand{\calD}{{\cal D}}
\newcommand{\calC}{{\cal C}}
\newcommand{\calX}{{\cal X}}
\newcommand{\calY}{{\cal Y}}
\newcommand{\calZ}{{\cal Z}}
\newcommand{\calA}{{\cal A}}
\newcommand{\calB}{{\cal B}}
\newcommand{\calG}{{\cal G}}
\newcommand{\calI}{{\cal I}}
\newcommand{\calJ}{{\cal J}}
\newcommand{\calR}{{\cal R}}
\newcommand{\calK}{{\cal K}}
\newcommand{\calU}{{\cal U}}
\newcommand{\barx}{\overline{x}}
\newcommand{\bary}{\overline{y}}

\newcommand{\bone}{{\bf 1}}
\newcommand{\btwo}{{\bf 2}}
\newcommand{\bthree}{{\bf 3}}
\newcommand{\ba}{\boldsymbol{a}}
\newcommand{\bb}{\boldsymbol{b}}
\newcommand{\bp}{\boldsymbol{p}}
\newcommand{\bt}{\boldsymbol{t}}
\newcommand{\bu}{\boldsymbol{u}}
\newcommand{\bv}{\boldsymbol{v}}
\newcommand{\bw}{\boldsymbol{w}}
\newcommand{\bx}{\boldsymbol{x}}
\newcommand{\by}{\boldsymbol{y}}
\newcommand{\bz}{\boldsymbol{z}}
\newcommand{\br}{\boldsymbol{r}}
\newcommand{\bh}{\boldsymbol{h}}

\newcommand{\bA}{\boldsymbol{A}}
\newcommand{\bC}{\boldsymbol{C}}
\newcommand{\bD}{\boldsymbol{D}}
\newcommand{\bG}{\boldsymbol{G}}
\newcommand{\bH}{\boldsymbol{H}}
\newcommand{\bR}{\boldsymbol{R}}
\newcommand{\bS}{\boldsymbol{S}}
\newcommand{\bT}{\boldsymbol{T}}
\newcommand{\bP}{\boldsymbol{P}}
\newcommand{\bX}{\boldsymbol{X}}
\newcommand{\bY}{\boldsymbol{Y}}
\newcommand{\bZ}{\boldsymbol{Z}}
\newcommand{\dc}{d}
\newcommand{\ds}{d^{(s)}}
\newcommand{\Gs}{G^{(s)}}
\newcommand{\Gc}{G}
\newcommand{\NN}{\mathbb{N}}
\newcommand{\RR}{\mathbb{R}}
\newcommand{\QQ}{\mathbb{Q}}
\newcommand{\vs}{v^{(s)}}
\newcommand{\vc}{v}
\newcommand{\abs}[1]{\left\lvert #1 \right\rvert}
\newcommand{\norm}[1]{\left\lVert #1 \right\rVert}
\newcommand{\ceil}[1]{\lceil#1\rceil}
\newcommand{\Exp}{\EX}
\newcommand{\floor}[1]{\lfloor#1\rfloor}
\newcommand{\cei}[1]{\lceil#1\rceil}

\newcommand{\EX}{\hbox{\bf E}}
\newcommand{\var}{\hbox{\bf var}}
\newcommand{\prob}{{\rm Prob}}

\newcommand{\gset}{Y}
\newcommand{\gcol}{{\cal Y}}

\newcommand{\eqdef}{:=}

\newcommand{\otilde}{\widetilde{O}}
\newcommand{\bigO}[1]{{O\left( #1 \right)}}

\newcommand{\wgt}{\mathrm{wt}}

\newcommand{\hypermb}{{\tt Hyper-MatulaBeck}}
\newcommand{\edgedeg}{{\tt Compute-Hyperedge-Degrees}}
\newcommand{\ancestordescendant}{{\tt Compute-Ancestor-Descendant}}
\newcommand{\trianglebasedalg}{{\tt Compute-Triangle-Based-Patterns}}
\newcommand{\containedalg}{{\tt Compute-$2$-$3$}}
\newcommand{\starsalg}{{\tt Compute-Stars}}
\newcommand{\extstarsalg}{{\tt Compute-Extended-Stars}}
\newcommand{\openalg}{{\tt Compute-$\compCounts$}}
\newcommand{\computelid}{{\tt Compute-$\localIndegrees$}}

\newcommand{\Sec}[1]{\hyperref[sec:#1]{Section~\ref*{sec:#1}}} 
\newcommand{\Eqn}[1]{\hyperref[eq:#1]{(\ref*{eq:#1})}} 
\newcommand{\Fig}[1]{{Fig.\,\ref{fig:#1}}} 
\newcommand{\Tab}[1]{\hyperref[tab:#1]{Tab.\,\ref*{tab:#1}}} 
\newcommand{\Table}[1]{\hyperref[tab:#1]{Table\,\ref*{tab:#1}}} 
\newcommand{\Thm}[1]{\hyperref[thm:#1]{Theorem\,\ref*{thm:#1}}} 
\newcommand{\Fact}[1]{\hyperref[fact:#1]{Fact\,\ref*{fact:#1}}} 
\newcommand{\Lem}[1]{\hyperref[lem:#1]{Lemma\,\ref*{lem:#1}}} 
\newcommand{\Prop}[1]{\hyperref[prop:#1]{Prop.~\ref*{prop:#1}}} 
\newcommand{\Cor}[1]{\hyperref[cor:#1]{Corollary~\ref*{cor:#1}}} 
\newcommand{\Conj}[1]{\hyperref[conj:#1]{Conjecture~\ref*{conj:#1}}} 
\newcommand{\Def}[1]{\hyperref[def:#1]{Definition~\ref*{def:#1}}} 
\newcommand{\Alg}[1]{\hyperref[alg:#1]{Alg.~\ref*{alg:#1}}} 
\newcommand{\Clm}[1]{\hyperref[clm:#1]{Claim~\ref*{clm:#1}}} 
\newcommand{\Obs}[1]{\hyperref[obs:#1]{Observation~\ref*{obs:#1}}} 
\newcommand{\Rem}[1]{\hyperref[rem:#1]{Remark~\ref*{rem:#1}}} 
\newcommand{\Con}[1]{\hyperref[con:#1]{Construction~\ref*{con:#1}}} 
\newcommand{\Step}[1]{\hyperref[step:#1]{Step~\ref*{step:#1}}} 
\newcommand{\Assumption}[1]{\hyperref[assm:#1]{Assumption\,\ref*{assm:#1}}} 

\newcommand{\baseline}{\textsf{\textit{MoCHy-E}}}
\newcommand{\baselineadv}{\textit{Exact-adv}}
\newcommand{\ditch}{\textsf{DITCH}}

\title{Counting Hypertriangles through Hypergraph Orientations}

\author{
  Daniel Paul-Pena\\
  University of California, Santa Cruz\\
  \texttt{dpaulpen@ucsc.edu}
  \and
  Vaishali Surianarayanan\thanks{Part of this work was carried out while the author was at the University of California, Santa Cruz.}\\
  Yale University\\
  \texttt{vaishalisurianarayanan@gmail.com}
  \and
  Deeparnab Chakrabarty\\
  Dartmouth College\\
  \texttt{deeparnab@dartmouth.edu}
  \and
  C. Seshadhri\\
  University of California, Santa Cruz\\
  \texttt{sesh@ucsc.edu}
}
\date{}

\maketitle

\begin{abstract}
	Counting the number of small patterns  is a central task in network analysis. While this problem is well studied for graphs, many real-world datasets are  naturally modeled as hypergraphs, motivating the need for efficient hypergraph motif counting algorithms. In particular, we study the problem of counting hypertriangles - collections of three pairwise-intersecting hyperedges. These hypergraph patterns have a rich structure with multiple distinct intersection patterns unlike graph triangles.
	
	Inspired by classical graph algorithms based on orientations and degeneracy, we develop a new theoretical framework that generalizes these concepts to hypergraphs and yields provable algorithms for hypertriangle counting. 
    
    We implement these ideas in DITCH (Degeneracy Inspired Triangle Counter for Hypergraphs) and show experimentally that it is faster on all but one small instance, achieving speedups of up to 100x, and is 6-100x more memory-efficient than existing state-of-the-art methods.
\end{abstract}

\textbf{Code Availability:}
\\
The source code has been made available at \url{https://github.com/daniel-paul/DITCH}.

\section{Introduction}

Graph analytics is an important research area with direct impact on industry. Hypergraphs are generalizations of graphs in which each edge can contain more than two vertices, allowing them to naturally capture higher-order interactions. 
They have become the model of choice~\cite{BaCe+20,ToBlBaEl+21,AnCoPo+23,LeBuEl+25} to capture group interactions in real-world systems such as research collaborations~\cite{BeAbSc+18}, social networks~\cite{AnCoPo+23}, image analysis~\cite{HuQiSh+10,YuTaWa12}, data clustering~\cite{AmVeBe20,VeBeKl22}, and protein analysis~\cite{HwTaTi+08}. See Figure~\ref{fig:intro-hypergraph-ex} for an example co-authorship hypergraph. This growing use of hypergraphs makes their efficient management and analysis at scale increasingly important.

Graph analytics has seen significant advances in both theory and practice.
Systems like Neo4j, KùzuDB~\cite{sigmod/Salihoglu23}, and Spanner Graph use algorithmic techniques such as worst-case optimal joins and degeneracy-based graph algorithms.

Hypergraph analytics is less mature, though there are some systems and libraries with native hypergraph support, such as Hypergraphx~\cite{hypergraphx/compnet/LotitoBNMB26} and TypeDB~\cite{dorn2024typeql}.
There is a recent trend of hypergraph-processing papers (we refer to the survey~\cite{LeBuEl+25}).
There is a major opportunity for algorithms and data management to come together for hypergraph analysis. 

\begin{figure}
 	\centering
 	\includegraphics[width=0.75\linewidth]{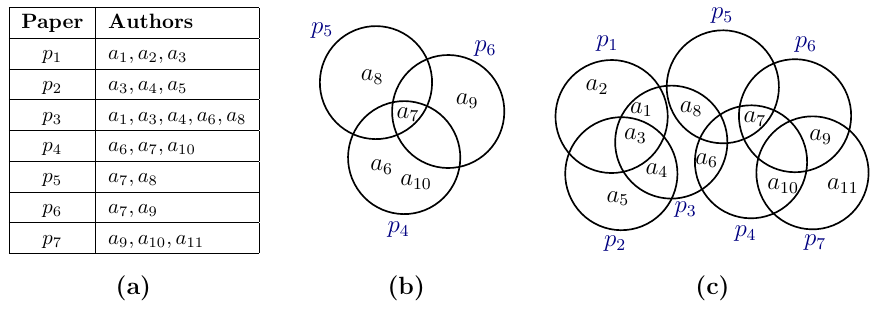}
    \caption{\small (a) A co-authorship dataset and (c) its hypergraph model, where vertices denote authors and hyperedges correspond to papers. The highlighted hypertriangle in (b) captures three papers sharing a common author, with all other authors unique to individual papers.}
 	\label{fig:intro-hypergraph-ex}
\end{figure}
 
A well-studied structural feature in graphs is the frequency of small patterns, also known as motifs. These have been used in numerous applications, including community detection, anomaly detection, and feature construction for graph learning models. 
The motif counting problem in graphs has received significant attention in both theory and practice~\cite{Lo67, DiSeTh02, FlGr04, DaJo04, Lo12, AhNeRo+15, CuDeMa17, PiSeVi17, SeTi19, CuNe25, DoMaWe25} (see tutorial~\cite{SeTi19}). Even the simplest non-trivial case, triangle counting, has developed into a rich line of research~\cite{BaKuSi02,JoGh05,ScWa05,ScWa05b,BuFrLe+06,Co09,SuVa11,JhSePi13,PaTaTi+13,KoPiPl+14,McVoVu16,BeCh17,TuTu19}, as triangles capture fundamental notions of transitivity, clustering, and local density in networks. Efficient algorithms for counting and listing motifs are an essential tool
for network analysis.

Considering the importance of graph motifs, researchers have recently pursued analogous notions in hypergraphs. Motif counts in hypergraphs have been used for applications such as distinguishing real-world hypergraphs from randomized counterparts~\cite{LeJiSh20,NiAmAkSa24} and understanding how patterns evolve over time~\cite{LeJiSh20}. There has also been a growing line of work~\cite{LeJiSh20,tkde/ZhangZWYK23,YiWaZh+24, MeYuLi+25} on counting the occurrences of motifs in hypergraphs.

Even simple patterns such as triangles admit multiple interpretations in hypergraphs, including edge-based and vertex-based definitions. In this work, we focus on one such definition that is particularly useful in applications. We study 26 hypertriangle motifs defined by three hyperedges and their intersection properties (defined in \hyperref[subsec:defns]{Section~\ref*{subsec:defns}}, Figure~\ref{fig:types}). We design new algorithms for counting these hypertriangles with theoretical guarantees. 
Our implementations are faster on all but one small instance, achieving \textbf{speedups of up to $\mathbf{100}\times$}, and are \textbf{$\mathbf{6\textbf{--}100}\times$  more memory-efficient} than all evaluated baselines (see Figure~\ref{fig:comparison}).

To the best of our knowledge, prior work on hypertriangle counting has not compared with general-purpose database systems and hypergraph analytics libraries. We additionally evaluate DuckDB and Hypergraphx~\cite{hypergraphx/compnet/LotitoBNMB26} alongside state-of-the-art hypertriangle algorithms. With carefully structured queries, DuckDB performs comparably to specialized hypertriangle-counting methods and, in some cases, is even more memory-efficient.

A key data-management challenge in hypergraph analytics is the blow-up of intermediate data, even when the input itself fits in main memory. When expressed as database queries, hypertriangle counting gives rise to cyclic joins with aggregation whose intermediate results can exceed main memory. We observe this behavior for existing hypertriangle baselines as well as DuckDB and Hypergraphx.
Our algorithms avoid this bottleneck, leading to huge memory savings and the ability to process two additional large datasets containing trillions of hypertriangles. We provide a detailed memory analysis showing how our data structures avoid this intermediate blow-up in Section~\ref{sec:memory}.

We now highlight a few applications of hypertriangle counts. 

\smallskip
\noindent{\bf Example 1.}
Consider the co-authorship hypergraph obtained from the DBLP database. 
Hypertriangle counts enable a {\em finer-grained analysis of collaboration patterns}, well beyond what is captured by graph representations.

We observe that about 31\% of triples of intersecting papers have a common author across all three. Among the triples with a common author across all three papers, a single hypertriangle (pattern 9 of \Fig{types}) accounts for about 84\% of such instances. This pattern captures triples of papers where authors are either common to all three papers or belong to exactly one paper. In 
\Fig{intro-hypergraph-ex}(b), edges $p_4,p_5,p_6$ form this pattern. While 99\% of these triples share exactly one common author, we discover about 30 explicit patterns with at least 10 common authors in the intersection region. (Details in \Sec{fine-grained}.) More such hypergraph analysis using hypertriangle counts have been done in \cite{LeJiSh20, YiWaZh+24}.

\smallskip
\noindent{\bf Example 2.}
Hypertriangle counts can serve as {\em informative features for learning tasks on hypergraphs.} For instance, in a co-authorship network such as DBLP, one may aim to predict future collaborations (i.e., hyperedges) based on past data. In biological or protein interaction networks, one may similarly aim to predict missing or future interactions. Prior work~\cite{LeJiSh20} models this as a binary classification problem distinguishing real hyperedges from synthetic ones.

\smallskip
\noindent{\bf Example 3.} 
Hypertriangle counts can also be used to {\em characterize and compare the structure of real-world hypergraphs across different domains.} By analyzing the distribution of hypertriangle patterns, one can identify domain-specific structural signatures, enabling effective distinction between datasets from different application areas. This is illustrated in Figure~\ref{fig:intro-similarity-matrix} (sourced from \cite{LeJiSh20}). 

\begin{figure}[ht]
 	\centering
 	\includegraphics[width=0.5\linewidth]{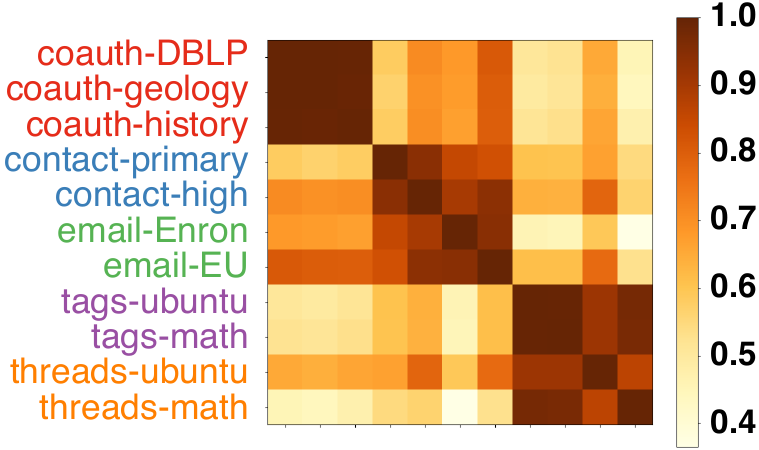}
    \caption{\small (Source~\cite{LeJiSh20}) Similarity matrix based on hypertriangle counts across datasets. Both rows and columns correspond to datasets (e.g., co-authorship, contact, email, tags and threads), and each entry denotes the Pearson correlation coefficient of the adjusted hypertriangle count distributions. The matrix reveals domain-level clustering, indicating that hypertriangle counts capture domain-specific structural patterns.}
 	\label{fig:intro-similarity-matrix}
\end{figure}

\smallskip
    We need algorithms that can compute hypertriangle counts efficiently in both time and memory. In graphs, such efficiency is achieved using structural tools such as vertex orderings and edge orientations that significantly reduce the search space. However, existing methods for hypergraph motif counting do not exploit such tools, leaving a significant gap in how to design efficient algorithms for hypergraphs.

In this work, we show that hypergraphs admit natural analogues of these tools. We define and use vertex orderings and hypergraph orientations to reduce the search space, leading to algorithms that are both theoretically sound and practically fast.

\begin{figure}[ht!]
	\centering
	\begin{subfigure}[t]{.5\textwidth}
		\centering
		\includegraphics[scale = 0.2]{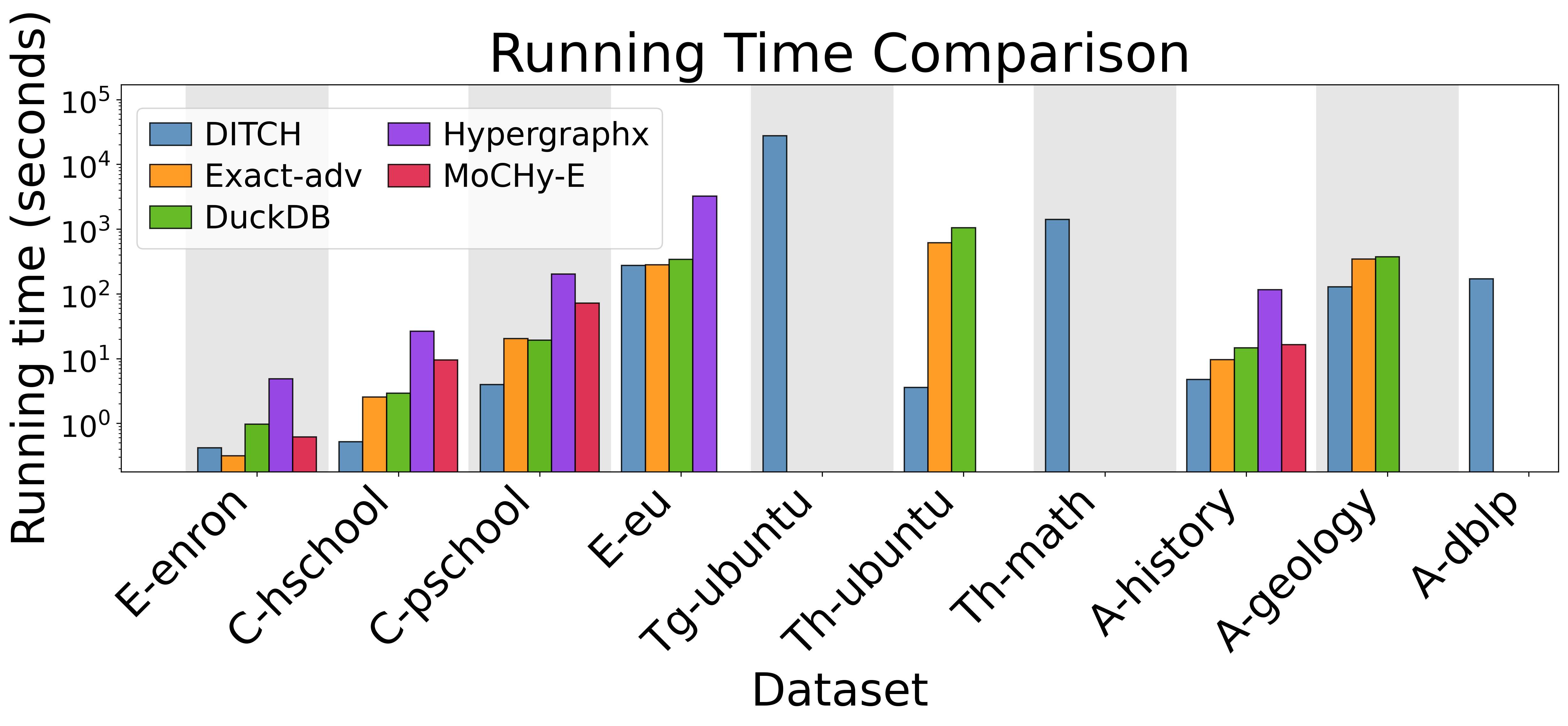}
		\caption{Running time comparison}
		\label{fig:runtime}
	\end{subfigure}
	\hspace{0.03\textwidth}
	\begin{subfigure}[t]{.5\textwidth}
		\centering
		\includegraphics[scale = 0.2]{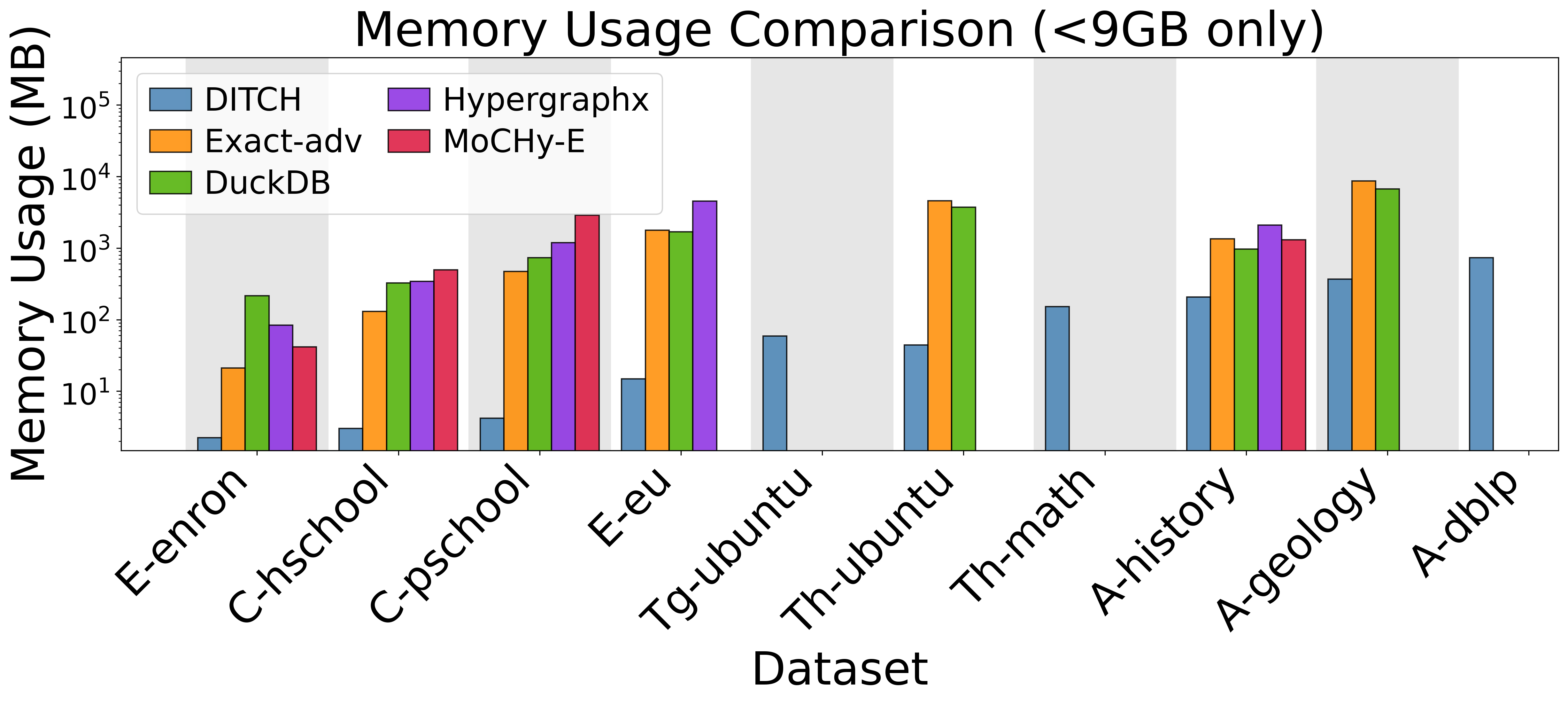}
		\caption{Memory comparison}
		\label{fig:memory}
	\end{subfigure}%
    \caption{\small Runtime and memory comparisons: we compare our algorithm \ditch{} with the existing best
    \baseline~\cite{LeJiSh20} and state-of-the-art \baselineadv~\cite{YiWaZh+24} algorithms as well as DuckDB and Hypergraphx. We test on
    a collection of ten datasets, and set a memory limit of 9GB. Any bar that is not present
    means that the algorithm ran out of memory. Across all but one dataset, \ditch{} is faster by up to 100x, and is 6-100x more memory-efficient across all datasets.}
	\label{fig:comparison}
\end{figure}

\smallskip
\noindent \textbf{Main contributions.}
\begin{asparaitem}
    \item \textbf{Structural tools for hypergraphs motifs.} We introduce vertex orderings and hypergraph orientations that generalize classical graph motif counting techniques and reduce the search space for hypergraph motif counting.

    \item \textbf{Efficient exact algorithms.} Using these tools, we design algorithms for counting all hypertriangle patterns with strong theoretical guarantees, and implement them in \ditch{}, which can also support listing several patterns.

    \item {\textbf{Experimental validation and analysis.} We show that \ditch{} outperforms existing state-of-the-art methods significantly in both runtime and memory on real-world datasets
    and enables additional fine-grained analysis. 
    We also compare against DuckDB and Hypergraphx, finding DuckDB competitive with specialized methods, while \ditch{} outperforms all baselines.}
\end{asparaitem}

\subsection{Definitions and Problem Statement}
\label{subsec:defns}
A hypergraph $G = (V, E)$ represents a set system, where each hyperedge $e \in E$ is a subset of $V$.
Each hyperedge $e \in E$ is a subset of vertices and $|e|$ denotes the number of vertices in $e$. No hyperedge is repeated. 
We use $n = |V|$ for the number of vertices, $m=|E|$ for the number of hyperedges and $\inputsize = n+\sum_e |e|$ for the input size. 
The \emph{rank} of a hypergraph, denoted $r$, is the maximum size of a hyperedge.

A {\em hypertriangle}, as defined by~\cite{LeJiSh20}, is a collection $\{e_1, e_2, e_3\}$of three hyperedges 
such that $e_i \cap e_j \neq \emptyset$ for all $i,j \in \{1,2,3\}$. This generalizes the notion of a triangle in a (pairwise) graph, however, 
unlike graphs there are multiple possible configurations of a hypertriangle depending on how these hyperedges intersect. 
In particular, given $\{e_1, e_2, e_3\}$, there are $7$ intersection subsets or ``regions'' as in a Venn diagram, and after excluding symmetric
and duplicated pattern from the $2^7$ possible configurations, we end up with $20$ {\em $h$-motifs}~\cite{LeJiSh20} which are enumerated in 
\Fig{types} as patterns $1$ to $20$; colored areas indicate non-emptiness. Patterns $21$ to $26$ are called ``open patterns'' and correspond to length $3$ paths in hypergraphs. In \Fig{intro-hypergraph-ex}(b), we show an explicit example of pattern 9. The hyperedges $\{p_4, p_6, p_7\}$ for an instance of pattern 19, since $p_4 \cap p_6 \cap p_7$ and $p_6 \setminus (p_4 \cup p_7)$ are empty.

\begin{figure}[htbp]
	\centering
	\includegraphics[width=0.6\linewidth]{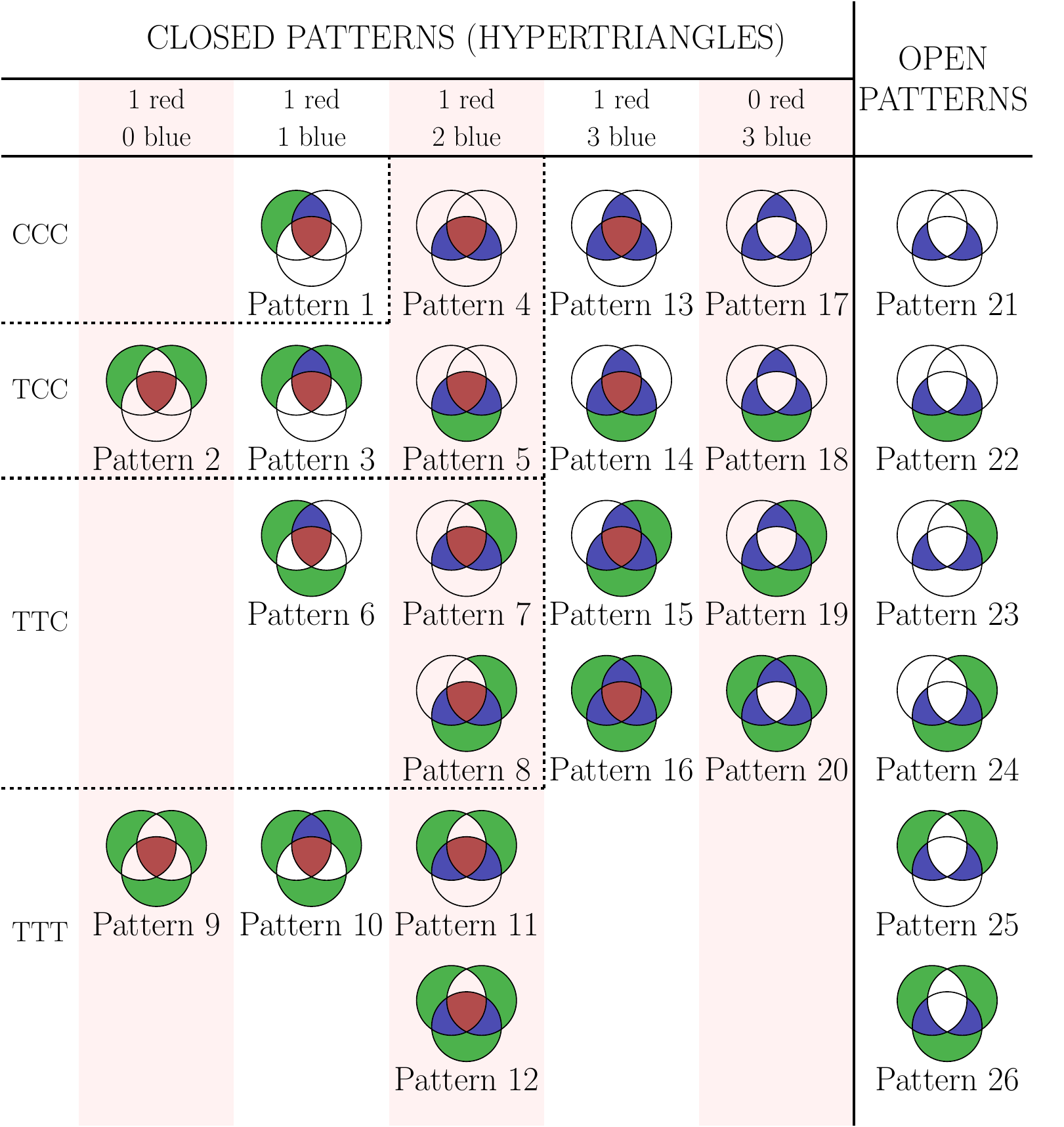}
    \caption{\small Hypertriangles, as defined by~\cite{LeJiSh20}. The three circles represent the distinct hyperedges,
    and regions are colored if they are non-empty. Regions in exactly one set are colored green,
    two sets are colored blue, and three sets are colored red. The 20 closed patterns are technically the hypertriangles,
    whereas the 6 open patterns are akin to paths. Following~\cite{YiWaZh+24}, we use the letter $C$ indicates inclusion while $T$ indicates intersection;
    so the lone pattern $1$ in CCC forms a chain $e_1 \subset e_2 \subset e_3$, while none of the edges in TTT patterns are contained in any other.}
	\label{fig:types}
\end{figure}

With the above definitions in place, we can clearly state the algorithmic problem we address in this paper.

\begin{problem}{Hypertriangle Counting Problem}
    \textbf{Input:} A hypergraph $G$. \\
    \textbf{Output:} The quantity $\ct{i}$ for $i \in [1,26]$ where $\ct{i}$ is the number of copies of motif $i$.
\end{problem}

We note here that except for patterns 6, 9, and 10 (and the open patterns), our algorithms can list the patterns in the same time as counting. The counts
of patterns 6, 9, and 10 are in the order of many billions, so our algorithm
determines the count without list. Nonetheless (demonstrated in \Sec{fine-grained}), our algorithm can compute numerous statistics on pattern (say) 9, and explicitly list instances with large intersection.

\subsection{Our Contributions Explained} \label{sec:contribution}
\noindent
{\bf Previous works.} 
The exact counting algorithm {\sf MoCHy-E} in~\cite{LeJiSh20} projected the hypergraph into a (normal) weighted graph whose vertices corresponded to hyperedges of the original hypergraph.
This leads to a large memory footprint and a running time of roughly $\sum_v d^3_v$ where $d_v$ is the number of hyperedges $v$ is present in. 
In~\cite{YiWaZh+24}, the main idea was to consider the hypertriangles as a collection of three ``hyper-wedges'', where a hyperwedge is a collection of two hyperedges which intersect. Although this leads to an empirical improvement in their implementation {\sc Exact-adv}, since it cuts down on redundant checks made in the previous algorithm, it doesn't escape going over {\em all} pairs of intersecting edges $e_i$ and $e_j$, and then searching over wedges that these particular ones participate in. 
This still leads to a $\sum_v d^3_v$ runtime in the worst case.\medskip

\noindent
{\bf Conceptual Contribution: Vertex Orderings.} Our main conceptual contribution is to use {\em vertex orderings} and {\em hypergraph orientations} which cuts down the search space both theoretically and, as we exhibit, empirically by orders of magnitude. In particular, given any ordering $\pi$ of the vertices $V$, we can {\em orient} every hyperedge $e$ by identifying the $v\in e$ with the largest $\pi(v)$ as the ``sink'' of $e$; 
we say that {\em out-degree} $d^+_v$ of a vertex $v$ is the number of edges $e$ containing $v$ where $v$ is not the sink. 
This identification of an edge with its sink is {\em crucial} in cutting down the search-space for motif-counting, as one can go over vertices and consider only their out-neighborhoods. This idea is inspired by the paper of Chiba and Nishizeki~\cite{ChNi85} who described an algorithm for triangle counting using graph orientations. \smallskip

\noindent
{\bf Theoretical Contribution: Minimum Degeneracy Ordering.} 
Since the out-degree of a vertex depends on the orientation and we wish the out-neighborhood to be as small as possible, 
the best ordering is one which minimizes the {\em maximum} out-degree. We give a linear time algorithm to find this ordering. Furthermore, we prove that this ``min-max out-degree'' is precisely equal to a notion of {\em hypergraph degeneracy} $\kappa(G) = \kappa$. This parameter is always smaller than the maximum degree, but is often orders of magnitude smaller (see~\Table{summary-datasets}.). Furthermore, it generalizes the notion in graphs, and the study of algorithms in bounded degeneracy graphs is a rich area.

There are some theoretical notions of hypergraph degeneracy~\cite{BrBrDe+25,PaSe26}, but those definitions
only apply to patterns defined by vertices.
For edge-based patterns like hypertriangles, we introduce
a new notion of hypergraph degeneracy.\smallskip

\noindent
{\bf Theoretical Contribution: Order-based Counting Algorithms.} 
Based on the above definitions, we design a host of theoretical algorithms that count all hypertriangles efficiently. 
Although we describe the main ideas behind our algorithms in \Sec{ideas}, the key insight is the one mentioned above: one only needs to inspect the ``outward'' hyperedges from a vertex,
vastly reducing the search space. In particular, we design algorithms
whose runtimes are comparable to $\sum_{v\in V} d_v(\dout{v})^2$ (as opposed to $\sum_v d^3_v$ in previous works), and thus if the hypergraph has bounded degeneracy, we obtain 
{\em linear time} algorithms matching the situation for triangle counting algorithms in graphs~\cite{ChNi85}.

The most compelling effect of the orientation is not just that the maximum out-degree is small, but rather we have a \emph{``taming of the tail"} effect. In \Fig{deg-dis},
we plot the degree and outdegree distributions for various datasets. We see how the outdegree tail is significantly
thinner than the heavy tail of the degree distribution. The maximum value is also smaller, typically 3-5
times in all datasets (\Tab{summary-datasets}). Since the running time of our algorithm replaces $d^2_v$ by $(\dout{v})^2$,
such a reduction leads to a ten-fold reduction in running time.
 \smallskip

 \begin{figure}
 	\centering
 	\includegraphics[width=0.65\linewidth]{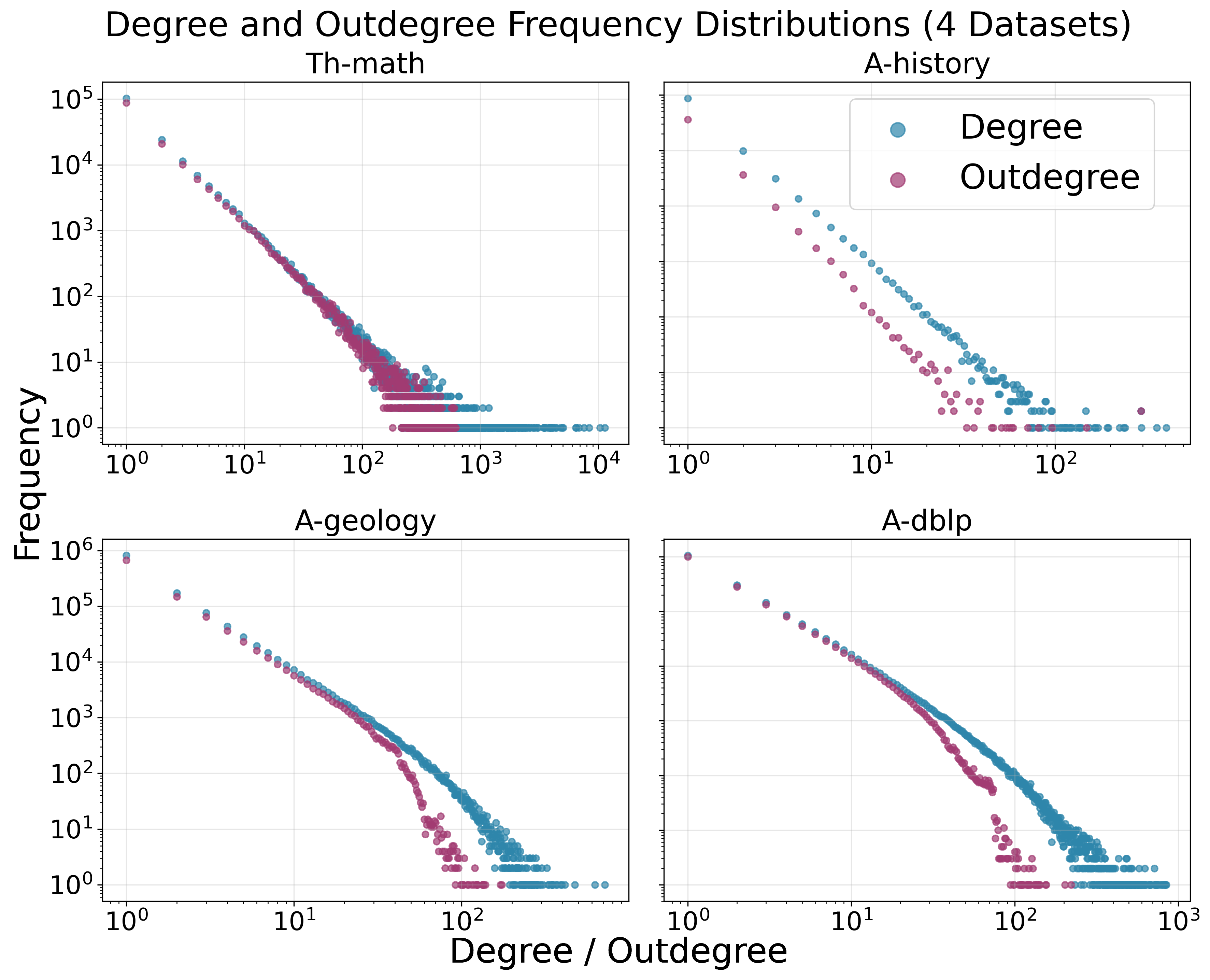}
 	\caption{\small Degree and Outdegree Distributions: For four datasets, we plot the difference between
 		the standard degree distribution and outdegree distribution of the oriented hypergraph. Each point on the $x$-axis
 		is a degree, and the $y$-axis is the frequency of that degree, plotting in log-log. Observe how the outdegree frequencies
 		are significantly smaller for higher degrees, which cuts down the search space for hypertriangles. For the co-author datasets
 		marked ``A'', the reduction of the tail is dramatic.}
 	\label{fig:deg-dis}
 \end{figure}

\noindent
{\bf Empirical Contribution.} We implement our theoretical algorithms
as part of our DITCH package on a simple commodity machine, and run them on a large variety of public hypergraph datasets.
We do comparisons with past algorithms~\cite{LeJiSh20,YiWaZh+24}.
As noted in~\Fig{comparison}, DITCH is at least \emph{10 times faster} (typically running in minutes) and uses {\em 10 times less memory} than the previous algorithms. 
For many of the datasets, previous algorithms run out of memory (and revert to approximation algorithms) while DITCH stays
under 1GB for 8 out of the 10 datasets, and under 4GB overall.

\subsection{Challenges and Algorithmic Ideas} \label{sec:ideas}

In this section we describe the challenges one faces in hypergraph motif-counting which aren't present in graph motif-counting. 
Subsequently, we describe how we address these challenges in our algorithms, and how the minimum degeneracy ordering helps us in speeding things. 

\begin{asparaitem}

\item[-] {\it The richness of hypertriangles:} In the graph setting, a connected pattern of 3 edges
can either be a triangle, a star, or a path. The latter two can be counted \emph{without}
enumeration by some combinatorial formulas. For hypergraphs, as \Fig{types} shows, there are many
possible cases of hypertriangles.  Even the basic wedge enumeration procedure for triangles is quite complex
for hypertriangles, as is shown in~\cite{YiWaZh+24}.

\item[-] {\it Edge-centric view as opposed to vertex-centric view.} 
Moreover, for graphs, we can entirely adopt a vertex-centric view,
which leads to simple formulas. For example, the number of 3-stars is simply $\sum_v {d_v \choose 3}$ where $d_v$ is the degree of vertex $v$.
In hypergraphs, the intersection areas (the red/blue regions) can
have different number of vertices for different hypertriangles. So one cannot devise easy formulas,
and naively one is forced to potentially enumerate over all $\sum_v {d_v \choose 3} = \Theta(\sum_v d^3_v)$ edge triples to classify
the hypertriangle. Essentially, all previous work has a running time of at least this amount. Given the heavy-tailed
degree distribution, the third moment is quite large. \smallskip
\end{asparaitem}

\noindent
{\bf Main Algorithmic Ideas.}
Let us now discuss a number of algorithmic ideas to get {\em all} the hypertriangle counts. 
We first discuss counting the more interesting closed patterns (the patterns 21-26 are open patterns).

{\it Idea 1, hyperedge degree computation:} For any hyperedge $e$, its degree is the number of (other) hyperedges that intersection with $e$.
This quantity is useful for computing various motif counts. For graphs, the degree of $e = (u,v)$
is just $(d_u + d_v - 2)$. So all edge degrees can be computed in linear time.
For hypergraphs, such a formula does not exist, and a direct algorithm would require explicitly inspecting
all intersecting hyperedges. We show that orientations help in computing \emph{hyperedge degrees} in $O(\rank\sum_v d_v d^+_v)$ time. This is a fundamental computation which goes beyond motif counting. Details can be found in~\Sec{degrees}.

{\it Idea 2, outwedge enumeration:} This is similar to the idea in~\cite{YiWaZh+24}, but sped up significantly due to the hypergraph orientations; indeed, this can be thought of as the hypergraph analog of the triangle counting idea of Chiba-Nishizeki~\cite{ChNi85}. To illustrate this, 
consider patterns 13 to 20 in \Fig{types}. Observe that they all have three blue regions; meaning,
for every pair of edges, there is a vertex $v$ in each pair's intersection that is not in the third edge (this is exactly like a triangle in standard graphs).
So we can enumerate over pairs of outneighbors of all vertices, and check whether this  hyperwedge participates in a triangle using another outneighborhood enumeration (for another vertex). 
This gives a running time of $O(\rank^2\sum_{v} d_v(d^+_v)^2)$. With some tweaks, this idea can also enumerate all hypertriangles with two blue regions (third column in \Fig{types}). Details can be found in~\Sec{triangle}.

{\it Idea 3, containment:} When there are fewer than two blue regions, we do not have the ``triangle like''
structure discussed above. But we observe for patterns 1, 2, and 3, one of the hyperedges is fully contained in another.
Once again, hypergraph orientations can be useful in exploiting this; in particular, we show that in $O(\rank \sum_{v} d_v(d^+_v)^2)$ time 
we can enumerate all pairs of hyperedges where one is contained in the other. From these pairs, patterns 1, 2, and 3 can be counted directly. Details in~\Sec{contain}.

{\it Idea 4, computing linear combinations:} At this point, only patterns 6, 9, and 10 are left. In fact, we note that all the counting algorithms above can be converted to listing 
since we find the count by incrementing whenever we see a hypertriangle of the correct class. For the patterns 6,9, and 10, however, we use a novel idea.
Instead of dealing with them directly, we compute three (independent) linear combinations
of many pattern counts. From the previous counts and these combinations, the remaining
counts can be computed. For instance, a slight twist on the containment-based algorithms for patterns 1-3 can be applied to get a linear
combination involving pattern 6 (see \Clm{type_six}.).

Patterns 9 and 10 are where it gets interesting. We can think of pattern 9 as the equivalent of a ``star''. Unfortunately,
the number of stars is \emph{not} bounded by $O(m\cdot \poly(\edegen))$, so a direct enumeration is not feasible.
Simple formulas like $\sum_v {d_v \choose 3}$ fail because an individual pattern is overcounted by the number
of vertices in the red region. Fixing this overcount directly would require an enumeration, which previous algorithms
do. Our main idea is the following. By spending $\sum_v (\dout{v})^3$ time, we can enumerate all triples of outedges
shared by a vertex, but that misses some instances of Pattern 9. These ``misses'' can be captured by a combinatorial
formula, but that includes instances of other patterns. With a careful case analysis, we prove that this formula
(together with other enumeration) can count the \emph{sum} of counts from patterns 1 to 16 (anything with a red region) in $O(\rank \sum_{v} d_v(d^+_v)^2)$ time.

A more intricate version
of the above argument counts a weighted sum of patterns 1 to 16, where the weight is exactly the number
of blue regions. This argument is an involved analysis that inspects the intersection regions of outwedges,
and adds a count depending the structure and indegrees of vertices in this region.
Overall, we compute independent linear combinations involving patterns 9 and 10, and so their counts can be inferred. 
Details of these ideas can be found in~\Sec{star}.

{\it Open patterns:} The patterns 21 to 26 technically are not hypertriangles but are, rather, ``open patterns'' and a generalization of 3-paths. 
We use the orientation ideas to develop algorithms for counting these, but they are deferred to Appendix~\ref{sec:open}.
Our experiments (and reports of runtimes and memories) include counting these as well.
\\ \\
All together, we obtain the following theoretical guarantee addressing the main problem studied in this paper.

\begin{mdframed}[skipabove=0.5\baselineskip, skipbelow=0.5\baselineskip,
                 innertopmargin=4pt, innerbottommargin=6pt]
\setlength{\parindent}{0pt}
\begin{theorem}
Given any hypergraph $G$, the counts of all 26 patterns can be computed in time $O(\rank^2 \sum_{v}d_v(\dout{v})^2) = O(\inputsize \rank^2 \edegen^2)$ and space $O(\inputsize)$.
\end{theorem}
\end{mdframed}

\section{Related work}

{\bf Applications of Hypergraphs:}

Hypergraphs represent group interactions and are widely used in various domains such as 
social networks~\cite{ShSuCh24}, 
bioinformatics~\cite{HwTaTi+08}, 
recommendation systems~\cite{WaChRo22}, 
image retrieval~\cite{HuQiSh+10}, 
and more~\cite{YuTaWa12,BeAbSc+18,AmVeBe20,VeBeKl22,LeBuEl+25}. 

 Thus, considerable efforts have been made to identify fundamental properties of real-world hypergraphs (see related work section in \cite{icdm/KookKS20}). For instance, Kook et al.~\cite{icdm/KookKS20} observed that real-world hypergraphs exhibit heavy-tailed degree, hyperedge size, and intersection size, as well as temporal patterns like diminishing overlaps, densification, and shrinking diameter. We contribute by introducing a new notion of orientations with low outdegrees, and show that real-world hypergraphs have orientations with significantly smaller outdegrees as compared to vertex degrees.

{\bf Hypergraph motif counting:}

Triangles are the most fundamental motifs, and in hypergraphs they have been studied as vertex-based and hyperedge-based variants. Vertex-based hypergraph triangles were first introduced as inner and outer triangles by Zhang et al.~\cite{tkde/ZhangZWYK23}: inner triangles are induced by three vertices contained in a single hyperedge, whereas outer triangles consist of three vertices that are pairwise co-contained in three distinct hyperedges. In the same work,~\cite{tkde/ZhangZWYK23} also proposed reservoir-based sampling schemes for counting such vertex-based triangles. Recently, Meng et al.~\cite{MeYuLi+25} extended this taxonomy by introducing hybrid triangles, where three vertices lie in one hyperedge and at least two of them also lie in another hyperedge. They use this richer classification to design streaming algorithms for both vertex- and edge-based triangles that adapt the effective sample size to the available memory.

Hyperedge-based triangles (also called hyper-triangles in past work) consist of three mutually intersecting hyperedges and were independently proposed by Zhang et al.~\cite{tkde/ZhangZWYK23} and by Lee et al.~\cite{LeJiSh20}. Lee et al. were the first to systematically classify such patterns by the emptiness or non-emptiness of the seven regions in the three-way intersection; this is the notion studied in this work. They also developed an algorithm, \baseline, for counting all the patterns according to the classification. In a more recent journal version~\cite{vldb/LeeYKKS24}, they generalize this classification beyond the binary setting by introducing \emph{$3$-h motifs}, where each region is constrained to have zero, at most~$\theta$, or more than~$\theta$ vertices for a fixed threshold~$\theta$. Yin et al.~\cite{YiWaZh+24} subsequently developed faster algorithms (\baselineadv) for computing hyper-triangles that are closed patterns (see Figure~\ref{fig:types}) by enumerating hyperwedges. Later, Niu et al.~\cite{bigdataconf/NiuAAS24} found top-$k$ hyper-triangles by optimizing the sizes of intersection regions. We use \baseline\ and \baselineadv\ as baselines in our experiments. 

Past works use the distribution of hyper-triangle patterns and hypergraph clustering coefficient to distinguish domains: the distribution of hyper-triangles is stable within a domain but can differ across domains. Hyper-triangle statistics also serve as good features for hyperedge prediction~\cite{LeJiSh20,YiWaZh+24} and for deriving informative insights~\cite{bigdataconf/NiuAAS24}. More recently, hyper-triangle counts have been used for hypermotif prediction and hypergraph motif representation learning~\cite{kdd/AntelmiCVPPS25}. Other higher-order motifs, such as simplices/cliques, together with more details on hypergraph motif counting, are surveyed in recent work~\cite{motif-counting-survey-25}.

{\bf Degeneracy and Subgraph Counting.} The first paper to connect degeneracy with subgraph counting was the seminal work of Chiba-Nishizeki regarding clique counting~\cite{ChNi85}.
Many subsequent works have used the degeneracy to find fast algorithms for subgraph counting, both in theory and practice~\cite{Ep94,AhNeRo+15,JhSePi15,PiSeVi17,OrBr17,JaSe17,PaSe20} (See \cite{Se23}, \cite{peerj-cs/CalmazB24}, and \cite{motif-counting-survey-25} for surveys of recent results).

{\bf Hypergraph Counting in Theory.} The problem of counting homomorphisms in hypergraph has been studied mostly from the point of view of parameterized complexity, for hypergraphs with bounded \cite{DaJo04} and unbounded rank \cite{GrMa14}. More recently, Bressan et al. \cite{BrBrDe+25} studied parameterized algorithms for subhypergraph counting  while Paul-Pena and Seshadhri~\cite{PaSe26} studied the problem of finding linear time algorithm in degenerate hypergraphs.

\section{Preliminaries} \label{sec:prelims}

\begin{table}[htbp]
    \centering
    \small
    \begin{tabular}{ll|ll}
        \hline
        Symbol & Meaning & Symbol & Meaning \\
        \hline
        $n$ & Number of vertices & $\Nchildren{e}$ & Number of descendants of $e$ \\
        $m$ & Number of hyperedges & $\Nintersect{e}$ & Number of hyperedges intersecting $e$ \\
        $\inputsize$ & Input size, $n + \sum_{e} |e|$ & $\edegen(G)$ & \HyperDegen of $G$ \\
        $\rank(G),\rank$ & rank of $G$ & $\din{v}$ & in-degree of $v$ \\
        $|e|$ & arity of $e$ & $\dout{v}$ & out-degree of $v$ \\
        $d_{v}$ & Degree of the vertex $v$ & $\Nout{v}$ & out-neighborhood of $v$ \\
        $\edeg{e}$ & Hyperedge degree of $e$ & $\source{e_1\cap e_2}$ & The first vertex of $e_1 \cap e_2$ \\
        $\parents{e}$ & Set of ancestors of $e$ & $\sink{e}$ & The last vertex of $e$ \\
        $\children{e}$ & Set of descendants of $e$ & $\type(e_1,e_2,e_3)$ & Pattern type of triplet $e_1,e_2,e_3$ \\
        $\Nparents{e}$ & Number of ancestors of $e$ & & \\
        \hline
    \end{tabular}
    \caption{Summary of notation used throughout the paper.}
    \label{tab:notation}
\end{table}

We set up the terminology and notation for the paper. Refer to \Tab{notation} for a summary. 
We start with some hypergraph basics. In all our notation, we assume that the input hypergraph $G$
is fixed, so we do not include $G$ in the notation. We assume that $G$ is unweighted, so no edge
is repeated.
The {\it arity} $|e|$ of a hyperedge $e$ is the number of vertices that it contains, and for our theoretical results we assume this is at least $2$ for every hyperedge. 
Thus, the rank of $G$ is the maximum arity over all edges. We now define degrees.

\begin{asparaitem}
	\item $d_{v}$: The \emph{degree} of vertex $v$ is the number of hyperedges containing $v$.
	\item $\edeg{e}$: The \emph{hyperedge-degree} $\edeg{e}$ of hyperedge $e$ is the number of other hyperedges $e'$ that \emph{intersect} $e$, that is, $\edeg{e} = |\{e' \in E(G): e\cap e' \neq \emptyset \ \textrm{and}\ e' \neq e\}|$. 
\end{asparaitem}

Since hyperedges form a set system, we introduce notation for various intersection properties.
We use $e, e'$ to denote hyperedges.

\begin{asparaitem}
	\item If $e \subset e'$, then $e$ is a \emph{descendant} of $e'$. Also, $e'$ is an \emph{ancestor} of $e$.
	\item We use $\children{e}$ to denote the family of descendants of $e$, and $\parents{e}$ to denote the family of ancestors.
The sizes of these families are denote $\Nchildren{e}$ and $\Nparents{e}$ respectively.
	\item $\Nintersect{e}$: This is an important notation, that denotes the number of hyperedges intersecting $e$
that are neither ancestors nor descendants. So $\Nintersect{e} \eqdef \edeg{e} - \Nparents{e} - \Nchildren{e}$.
\end{asparaitem}

\smallskip

We note that the quantities $\Nchildren{e}, \Nparents{e}$, and $\Nintersect{e}$ are not obvious to compute.
Indeed, we will show in \Sec{degrees} how the correct hypergraph orientations help in computing them quickly.

\noindent
{\bf On the hypertriangle patterns:} The hypertriangle patterns are listed in \Fig{types}. 
We divide these patterns into closed ($1$ to $20$) and open ($21$ to $26$). In the closed patterns, the three hyperedges intersect pairwise, forming a hypertriangle; while in the open patterns, there is a hyperedge intersecting with the other two, which do not intersect with each other.

\begin{asparaitem}
        \item We use $\ct{i}$ to denote the count of the $i$th pattern.
\end{asparaitem}
For any three distinct edges $e_1,e_2,e_3$, let $\type(e_1,e_2,e_3)$ denote the pattern type induced by the three edges $e_1$, $e_2$, and $e_3$. 
Given three distinct edges $e_1,e_2,e_3$, the pattern type $\type(e_1,e_2,e_3)$ can be computed in $O(\rank)$ time.
We can just do a linear search through all the vertices in $e_1, e_2$, and $e_3$, and determine which
of the $7$ regions it belongs to.

\section{Hypergraph Orientations and Degeneracy} \label{sec:degen}

We begin with the notion of hypergraph orientations using the concept of {\em directed acyclic hypergraphs} (DAH) from~\cite{PaSe26}.
A \dah{} $\vec{G} = (G,\pi)$ is a pair where $G$ is a hypergraph and $\pi$ is an ordering of the vertices in $V(G)$. Each hyperedge $e$ can
be thought of as a sorted list $e_\pi$ of vertices that follow the ordering in $\pi$.
Typically, the ordering will be fixed, so we will not carry $\pi$ in the notation.
We will consider various subsets of vertices formed by intersections and unions of edges. Let $g \subseteq V$.
We give some crucial definitions on directed degrees.

\begin{asparaitem}
	\item The \emph{source} of $g$, denoted $\source{g}$ is the first vertex in $g$ according to the ordering $\pi$. 
	\item The \emph{sink} of $g$, denoted $\sink{g}$ is the last vertex in $g$ (according to $\pi$).
	\item The \emph{indegree} of vertex $v$ is the number of hyperedges for which $v$ is the sink. Formally, $\din{v} = |\{e \in E(G): \sink{e} = v\}|$.
	\item The \emph{outlist} of vertex $v$ is the list of hyperedges containing $v$ where $v$ is \emph{not} the sink.
	So $\Nout{v} = \{e \in E(G): v \in e, v\neq\sink{e}\}$. We denote the \emph{outdegree} of vertex $v$ as $\dout{v} = |\Nout{v}|$.
\end{asparaitem}

Given an ordering $\pi$, our algorithms have runtime as a function of these $\dout{v}$'s. It is therefore natural to consider the order
$\pi$ which minimizes the maximum $\dout{v}$. We next show how this is related to a {\em degeneracy} notion of the hypergraph, and 
how a generalization of the Matula-Beck algorithm~\cite{MaBe83} can compute this order in near linear time. 

\noindent
{\bf Degeneracy.} 
Given a hypergraph $G$ and a set of vertices $V' \subseteq V(G)$, the \multihyper{} of $G$ induced by $V'$, $\trimmed{G}{V'}$, is the multi-hypergraph $G'$ with vertex set $V'$ and hyperedge {\bf multiset} $E' = \{e \cap V' ~:~ e \in E(G), |e \cap V'| \geq 2\}$, where we keep multiple copies of $e\cap V'$ 
if they are the same. We stress that $G$ is a simple hypergraph, but $\trimmed{G}{V'}$ is not.
In a \multihyper, the degree of a vertex is the number of edges it participates in.
For convenience, we use $\delta(H)$ to denote the minimum degree of a vertex in the hypergraph $H$.

\begin{definition} [\HyperDegen]
	The \HyperDegen{} of $G$, $\edegen(G)$, is the maximum minimum degree across all \multihyper{s} of $G$.
	\[
	\edegen(G) = \max_{V' \subseteq V(G)} \delta(\trimmed{G}{V'})
	\]
\end{definition}

    This definition is different from theoretical notions
    in~\cite{PaSe26}, since we keep duplicated hyperedges when computing the trimmed subhypergraph. Therefore, our degeneracy will be larger in value. This is necessary in order to compute edge-centric patterns like hypertriangles.

The following theorem shows the connection with hypergraph orderings which minimize the maximum out-degree and fast algorithms to compute it.
\begin{theorem}\label{thm:dual}
	$\edegen(G) = \min_\pi \max_v \dout{v}$ and there is a $O(\inputsize)$ time algorithm to compute both the quantity
	and the ordering $\pi$ that minimizes the RHS.
\end{theorem}
First we show that any ordering $\pi$ upper bounds $\edegen(G)$, thus proving $\edegen(G)$ is {\em at most} the RHS in the above theorem.
\begin{lemma}\label{lem:ub}
	Fix any ordering $\pi$ of the vertices and let $\vec{G} = (G,\pi)$. 
	Let $x\in V$ be the vertex with the largest outdegree $\dout{x}$ in $\vec{G}$. Then, 
	$\dout{x} \geq \edegen(G)$.
\end{lemma}
\begin{proof}
	Let $V'\subseteq V(G)$ be the subset such that $\edegen(G)$ equals $\delta(\trimmed{G}{V'})$. Let
	$z = \source{V'}$ be the first vertex in $\pi$ present in $V'$; by definition, $\edegen(G)$ is at most the degree of $z$ in $\trimmed{G}{V'}$.
	Now, consider any edge $e\in E$ containing $z$ such that $e\cap V'$ is present in $\trimmed{G}{V'}$. Such an edge $e$ also must contain some other $y\in V'$ since $|e\cap V'| \geq 2$,
	and by our choice of $z$, this $y$ comes after $z$ in $\pi$. In particular, in $\vec{G}$, $e$ is an out-neighbor of $z$ since $z$ is not the last vertex.
	Thus, $\dout{z}$ is at least the degree of $z$ in $\trimmed{G}{V'}$. Thus, $\dout{z}\geq \edegen(G)$, and the claim follows.
\end{proof}

The next lemma shows that there is a vertex ordering where the maximum outdegree is exactly $\edegen(G)$, thus proving~\Thm{dual}.
Furthermore, this ordering can be computed in linear time. The algorithm is a generalization of a classic procedure of Matula and Beck~\cite{MaBe83} for graphs, and proceeds by peeling of vertex with minimum degree.

\begin{algorithm}[ht!]
	\caption{\hypermb$(G)$\\Input: Hypergraph $G$\\ Output: Degeneracy ordering $\pi$}\label{alg:degeneracy}
	\begin{algorithmic}[1]
		\State Initialize $\pi$ as an empty list
		\While{$G$ is not empty}
		\State Let $v$ be the minimum degree vertex in $G$
		\State Append $v$ to $\pi$
		\State Remove $v$ from $G$
		\For{$e \in E(G): e\supset v$}
		$e  = e\setminus v$
		\If{$|e|<2$}
		\For{$v' \in e$}
		\State $d_{v'} \gets d_{v'}-1$
		\EndFor
		\State Remove $e$ from $G$
		\EndIf
		\EndFor
		\EndWhile
		\State Return $\pi$
	\end{algorithmic}
\end{algorithm}

\begin{restatable}{lemma}{matulabeck} \label{lem:mb}
	For any input hypergraph $G$, Algorithm \ref{alg:degeneracy} computes an ordering $\pi$ of $G$ such that the resulting DAH $\vec{G} =(G,\pi)$, has maximum outdegree $\edegen(G)$. Moreover, the algorithm runs in time $O(\inputsize)$.
\end{restatable}

To implement this algorithm, we use a series of buckets that store the vertices based on their degree. We can find the minimum degree vertex in $O(h)$ total amortized time. On vertex removal, we update the size of the hyperedges containing it. If a hyperedge has size less than $2$ we remove it, and update the degree of its remaining vertex by moving its to the lower bucket.

As shown in \Fig{deg-dis}, the distribution of outdegrees obtained in \Lem{mb} has a significantly thinner tail than the 
the original outdegree distribution. The maximum outdegree $\edegen$ is also smaller than the maximum degree, and as seen in~\Tab{summary-datasets}, 
can be an order of magnitude smaller. More significantly, as we see, the runtimes of our algorithms are of the form $\sum_{v\in V} (\dout{v})^2$ or 
$\sum_{v\in V} (\dout{v})^3$ which are much more favorable to $\sum_{v\in V} d_v^3$ which is essentially the runtimes of the wedge-based counting methods from previous works.

\subsection{Fast hyperedge-degree computation} \label{sec:degrees}

We now show how the degeneracy orientations lead to an $O(mr^2\edegen)$ time algorithm to compute all the hyperedge-degrees $\edeg{e}$, which, recall, 
is the number of hyperedges that share a vertex with $e$.
It is non-trivial to compute the values of $\edeg{e}$, since we cannot use simple combinatorial formulas based on
vertex degrees.
Enumerating the pairs of hyperedges that share a vertex can be potentially an expensive procedure, and for a single hyperedge $e$ there are as many as $\sum_{v \in e} d_v$ such pairs. 
A direct implementation would have running time $O(\sum_v d^2_v)$. 
We show how to improve on that running time without explicit enumerating such pairs.

The basic idea is simple. Consider the neighboring edges $e'$ of $e$. 
Suppose the last vertex of $e'$, $\sink{e'}$, is not in $e$. To count such neighboring edges, we go over vertices $v\in e$ and enumerate the {\em out-neighbors} of $v$.
(We take care not to double count.)
Suppose $e'$ has its endpoint in $e$. Such edges can be counted by summing the in-degrees of $v\in e$ in $O(r)$ time (note that a hyperedge can only have one endpoint).

\begin{lemma} \label{lem:deg}
	Algorithm \ref{alg:degrees} computes the hyperedge-degrees of all the hyperedges in $G$ in time $O(\rank\sum_v d_v d^+_v) = O(\inputsize \rank \edegen)$.
\end{lemma}
\begin{proof}
	For every hyperedge $e$, the algorithm looks at each of the vertices on it, for each such vertex $v$, we further look at the out-neighbors $e'$ of $v$. 
	The number of pairs seen for $e$ is at most $\sum_{v \in e} d^+_v$. The total number of pairs seen is $\sum_e \sum_{v \in e} d^+_v$. Let $\bone_{v \in e}$ be the indicator
	for vertex $v$ being in $e$. We can write this sum as $\sum_e \sum_v \bone_{v \in e} d^+_v = \sum_v d^+_v \sum_e \bone_{v \in e} = \sum_v d_v d^+_v$.
	The final check on line $7$ can be performed in $O(r)$ time, by doing a linear search to find the common source of $e$ and $e'$, therefore the total complexity will be 
	$O(\rank \sum_v d_v d^+_v)$.
	
	We now prove correctness. Fix $e$, and let $e'$ be a different hyperedge. If $e\cap e' = \emptyset$, then $e'$ can not be either in-neighbor or out-neighbor of a vertex in $e$, therefore it can not contribute to the counts of $\edeg{e}$.
	
	Let $|e\cap e'|>0$. We distinguish two cases:
	
	\begin{asparaitem}
		\item $\sink{e'} \in e$: In this case $e'$ will contribute to $\din{v}$ for some $v$ in $e$, and therefore it will be counted in line $5$, note that $\sink{e'}\in e$, and therefore it can not be counted in line $8$.
		\item $\sink{e'} \notin e$: Let $v = \source{e\cap e'}$, only when the algorithm reaches line $7$ with $e$,$v$ and $e'$ will it count it in line $8$.\qedhere
	\end{asparaitem}
	
\end{proof}

\begin{algorithm}[t!]
	\caption{\edgedeg$(\vec{G})$\\Input: DAH $\vec{G} = (G,\pi)$\\ Output: Array of Hyperedge degrees $\edeg{\cdot}$}\label{alg:degrees}
	\begin{algorithmic}[1]
		\State Initialize $\edeg{\cdot}$ as an array containing $0$ for every hyperedge
		\For{$e \in E(G)$}
		\State $\edeg{e} \gets 0$
		\For{$v \in e$}
		\State $\edeg{e} \gets \edeg{e} + \din{v}$   \Comment{Case $(1)$}
		\For{$e' \in \Nout{v}\setminus\{e\}$}
		\If{$v = \source{e\cap e'}$ \textbf{and} $\sink{e'} \not\in e$}
		\State $\edeg{e} \gets \edeg{e} + 1$ \Comment{Case $(2)$}
		\EndIf
		\EndFor
		\EndFor
        \State $\edeg{e} \gets \edeg{e} - 1$ \Comment{Subtract self-count}
		\EndFor
		\State Return $\edeg{\cdot}$
	\end{algorithmic}
\end{algorithm}

Using these ideas, it is easy to enumerate all ancestor-descendant pairs. Recall that $e$ is a descendant of $e'$ and $e'$ is an ancestor of $e$ if $e\subset e'$, 
and $\Nparents{e}$ is the number of ancestors of $e$. We now show an algorithm to 
construct all these pairs.

\begin{algorithm}[t!]
	\caption{\ancestordescendant$(\vec{G})$\\Input: DAH $\vec{G} = (G,\pi)$\\ Output: Arrays containing the numbers of ancestors and descendants of each hyperedge, $\Nparents{\cdot},\Nchildren{\cdot}$}\label{alg:parents}
	\begin{algorithmic}[1]
		\State Initialize $\Nparents{\cdot}, \Nchildren{\cdot}$ as arrays containing $0$ for each hyperedge
		\For{$e_1 \in E(G)$}
        \State $v \gets \source{e_1}$
		\For {$e_2 \in \Nout{v}$}
		\If{ $e_2 \supset e_1$}
		\State $\Nparents{e_1} \gets \Nparents{e_1} + 1$
		\State $\Nchildren{e_2} \gets \Nchildren{e_2} + 1$
		\EndIf
		\EndFor
		\EndFor
	\end{algorithmic}
\end{algorithm}

\noindent
\emph{Remark.}
	This algorithm can also be used to explicitly construct the sets of ancestors and descendants of each hyperedge in the same time, by replacing lines $6$ and $7$. However these sets will take super-linear space and thus are computed
	on the fly when needed in \Sec{contain}.

\noindent
\begin{lemma} \label{lem:anc-des}
	\Alg{parents} computes the number of ancestors and descendants of each hyperedge in time $O(r \sum_{e}\dout{\source{e}})=O(m \rank \edegen)$.
\end{lemma}
\begin{proof} 
The first for iterates over all the edges, and the second over all out-neighbors of its source. Lines $5-7$ can be computed in $O(r)$ time. Giving $O(r \sum_{e}\dout{\source{e}}) = O(r \edegen \sum_e 1 = O(\inputsize \rank \edegen)$.

For the correctness, we need to verify that the algorithm will check every parent $e_2$ of every hyperedge $e_1$. Note that we are assuming $|e_1|>1$. Let $v = \source{e_1}$, if $e_2 \supset e_1$ we will have $e_2 \in N^+(v)$, and therefore the algorithm will check the pair $e_1,e_2$.
\end{proof}

We can then compute the variants of hyperedge degree: the number of ancestors (hyperedges that contain $e$) $\Nparents{e}$, the number of descendants $\Nchildren{e}$, and the number of intersections $\Nintersect{e}$. Note that $\edeg{e} = \Nparents{e} + \Nchildren{e} + \Nintersect{e}$. For convenience, we state the following corollary.

\begin{corollary} \label{cor:hyp-deg} The values of $\edeg{e}$, $\Nparents{e}$, $\Nchildren{e}$, $\Nintersect{e}$ can be computed for all edges $e$ in time $O(r\sum_v d_v d^+_v)$ $= O(\inputsize\rank\edegen)$.
\end{corollary}

\section{Counting All Hypertriangles} \label{sec:hypertri}

With the degeneracy orientation and all the hyperedge degree information, we can describe the DITCH suite of algorithms
that compute all pattern counts. There are multiple procedures, so the analysis is split into subsection
corresponding to each procedure. We follow the presentation in \Sec{ideas} which gave an overview of the techniques.

\begin{asparaitem}
	\item {\it Triangle based patterns, \Sec{triangle}:} These are the patterns formed by three hyperedges such that there are three vertices that form a classic graph triangle like structure.
        This approach counts all patterns with $3$ blue regions and two blue regions with one red region (the third to fifth column in \Fig{types}).
	\item {\it Containment based patterns, \Sec{contain}:} These algorithms handle patterns 1, 2, 3, and 6, which involve an ancestor-descendant pair of hyperedges.
The degree computation done in \Sec{degrees} is central for these algorithms.
    \item {\it Star patterns, \Sec{star}:} This section is entirely devoted to counting patterns 9 and 10. As discussed in \Sec{ideas},
        we design algorithms that compute linear combinations of counts of various patterns.
    \item {\it Open patterns:} These
        are technically not hypertriangles, and more akin to hyperpaths. Still, the orientation ideas can be applied to these
        patterns as well. Their discussion is deferred to Appendix~\ref{sec:open}.
\end{asparaitem}

\subsection{Triangle based algorithm} \label{sec:triangle}

As mentioned, the patterns $13$ to $20$ contain three blue regions. Note that, given an instance of one of these patterns, selecting a vertex from each of the blue regions will lead us to a ``triangle''-like structure. Where each pair of vertices is part of each of the hyperedges. We can therefore use wedge enumeration techniques to find all such triangles.

Consider the example on the top-left of \Fig{triangle}. The main idea is to start at a vertex $u$ and enumerate over each pair of hyperedges $e_1,e_2$ where $u$ is a non-sink. For each of these pairs we can then enumerate over their vertices ($u$ and $w$) and look at their out-neighborhoods again to find a hyperedge $e_3$ that ``closes'' the triangle. This is much faster than the previous algorithms, as we are looking at the outneighborhoods of vertices (the outdegree), rather than to all the hyperedges containing each vertex (the degree).

This idea also works for patterns with only $2$ blue regions and $1$ red region. As one can see in \Fig{triangle}, these patterns will also have three vertices following a similar property to the three blue regions, and therefore we can find them using the same algorithm.

We can improve the previous algorithm by introducing the concept of a witness. Basically for each triplet $\{e_1,e_2,e_3\}$ of edges corresponding to one of these hypertriangles, the witness $w(\{e_1,e_2,e_3\})$ is the sorted (according the degeneracy ordering) triplet of vertices forming a triangle with earlier ordering. Note that we can compute the three vertices of the witness in $O(|e_1|+|e_2|+|e_3|)$ time. The formal definition of the witness is given in \Def{witness}.

The witness allows us to shave one factor of the rank from the complexity of the algorithm, (because we do not iterate over all vertices in $e_3$). And more importantly, it can be used to avoid counting the same triplet multiple times. This is thanks to the check in line $8$ of \Alg{triangles}, which can only be satisfied once per triangle, as it forces $w'$ to be in the intersection of $e_2$ and $e_3$, and ensuring that $u$ and $v$ are the first two vertices of the witness.

\begin{figure}[htbp]
    \centering
    \includegraphics[width=.9\linewidth]{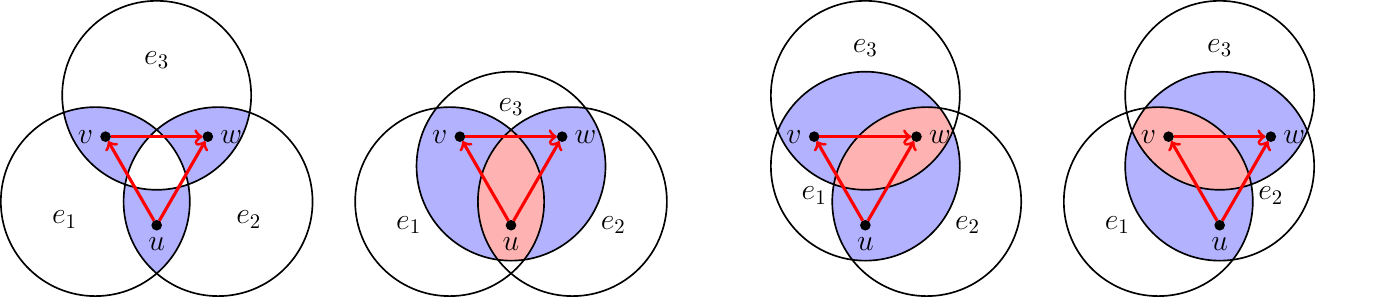}
    \caption{Examples of hypergraph patterns that are found using the Triangle based algorithm.}
    \label{fig:triangle}
\end{figure}

\begin{restatable}{lemma}{triangles} \label{lem:triangles}
    \Alg{triangles} determines the counts $\ct{4}, \ct{5}, \ct{7}, \ct{8}$, and each $\ct{i}$, $i \in [11,20]$ in time $O(\rank^2 \sum_{v} d_v(d^+_v)^2) = O(\inputsize\rank^2\edegen^2)$ 
\end{restatable}

\begin{algorithm}
\caption{\trianglebasedalg$(\vec{G})$\\
Input: DAH $\vec{G} = (G,\pi)$ \\
Output: The counts $\ct{4},\ct{5},\ct{7},\ct{8}$ and $\ct{i}$ for $i\in [11,20]$}\label{alg:triangles}
\begin{algorithmic}[1]
\For{$u \in V(G)$}
    \For{$e_1 \in \Nout{u}$}
        \For{$v \succ  u \in e_1$}
            \For{$e_2 \in \Nout{u}$, $e_3\in \Nout{v}$}
                \State $\type \gets \type(e_1,e_2,e_3)$
                \If{$\type \in[4,5]\cup[7,8]\cup[11,20]$}
                \State $(u',v',w') \gets w(e_1,e_2,e_3)$
                    \If{$u=u'$ \textbf{and} $v=v'$ \textbf{and} $w' \in e_2\cap e_3$}
                        \State $counts[\type] \gets counts[\type] + 1$
                    \EndIf
                \EndIf 
            \EndFor
        \EndFor
    \EndFor
\EndFor
\end{algorithmic}
\end{algorithm}

\subsection{Contained hyper-triangles} \label{sec:contain}

Recall $\Nchildren{e} = |\children{e}|$ and $\Nparents{e} = |\parents{e}|$. Note that~\Alg{parents} can compute these for all edges.
Given these amounts we can compute directly the number of instances of pattern $1$.

\begin{restatable}{claim}{typeone} \label{clm:type_one}
$\ct{1} = \sum_{e \in E(G)} \Nchildren{e} \cdot \Nparents{e}$.
\end{restatable}

For computing the counts of patterns $2$ and $3$ it suffices to iterate over pairs of ancestors of each hyperedge, and then directly check their intersections. This can be done efficiently using \Alg{double_parents}.

\begin{restatable}{claim}{typetwothree} \label{clm:two_three} \Alg{double_parents} correctly counts $\ct{2}$ and $\ct{3}$ in $O(\rank\sum_{v} d_v(\dout{v})^2) = O(\inputsize\rank\edegen^2)$ time.
\end{restatable}

\begin{algorithm}
\caption{\containedalg($\vec{G}$) \\Input: DAH $\vec{G} = (G,\pi)$ \\
Output: The counts $\ct{2}$ and $\ct{3}$}\label{alg:double_parents}
\begin{algorithmic}[1]
\For{$e_1 \in E(G)$}
    \For{$e_2,e_3\in \parents{e_1}$}
        \State $\type \gets \type(e_1,e_2,e_3)$
        \If{$\type \in [2,3]$} 
            \State $counts[\type] \gets counts[\type] + 1$
        \EndIf
    \EndFor
\EndFor
\end{algorithmic}
\end{algorithm}

Finally, as we mentioned we compute the number of copies pattern $6$ by computing a linear combination of the patterns $4$ to $8$. This can be done in constant time once the hyperedge degrees have been computed.

\begin{restatable}{claim}{typesix} \label{clm:type_six}
    $\ct{6} + \ct{7} + \ct{8} + 2 (\ct{4}+\ct{5})= \sum_{e} \Nparents{e} \cdot \Nintersect{e}$.
\end{restatable}

Altogether, we get the following lemma.
\begin{lemma} There is an algorithm that, given the counts $\ct{4}$, $\ct{5}$, $\ct{7}$, $\ct{8}$, computes $\ct{1}$, $\ct{2}$, $\ct{3}$ and $\ct{6}$ in additional $O(\rank\sum_{v} d_v(\dout{v})^2)$ $=$ $O(\inputsize \rank \edegen^2)$ time.
\end{lemma}
\begin{proof} 
    The number of ancestors and descendants together with the hyperedge-degrees can be computed in time $O(r\sum_v d_v d^+_v)$ using Algorithms \ref{alg:degrees} and \ref{alg:parents}. We can then compute patterns of type $1$ and $6$ in $O(m)$ time using \Clm{type_one} and \Clm{type_six} respectively.
    Finally, we use Algorithm \ref{alg:double_parents} for computing the number of patterns $2$ \& $3$.
\end{proof}

\subsection{Star-based counting} \label{sec:star}

All the hypertriangles (patterns $1$ to $20$) except for patterns $8$ and $9$ can be computed using the algorithms from the previous two sections. We will compute the total number of stars and extended-stars. These linear combinations will allow us to immediately get the number of copies of patterns $8$ and $9$.

\subsubsection{Star counts}

Any hypertriangle that has a red region (patterns 1 to 16) can be thought of as containing a ``star'',
since there is some vertex common to all hyperedges.

Let $e_1,e_2,e_3$ be three hyperedges forming a ``star'', as shown in \Fig{stars}. We can distinguish two cases:
\begin{asparaenum}
    \item One of the hyperedges $e_1$,$e_2$,$e_3$ has its sink in $e_1 \cap e_2 \cap e_3$. We can obtain the count of such stars by looking at the in-degree and out-degree of every vertex.
    \item $e_1$,$e_2$,$e_3$ do not have their sink in $e_1 \cap e_2 \cap e_3$. In this case we can list and count all such stars by starting at a vertex $v$ and looking at triplets of out-hyperedges from $v$, and verifying that $v$ is the source in $e_1 \cap e_2 \cap e_3$ (to avoid duplicate counts).
\end{asparaenum}
\begin{figure}[htbp]
    \centering
    \includegraphics[width=0.4\linewidth]{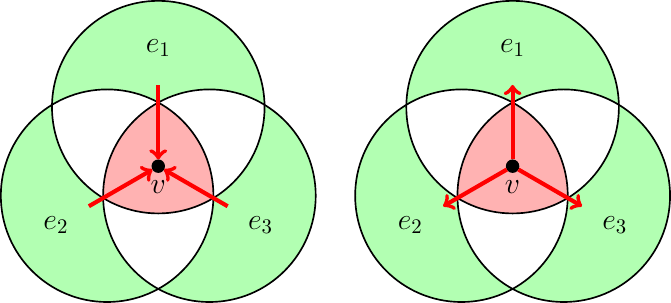}
    \caption{The two possible cases in a star-based hypertriangle.}
    \label{fig:stars}
\end{figure}
\begin{restatable}{lemma}{stars}\label{lem:star}
    \Alg{stars} computes the total $\sum_{i=1}^{16} \ct{i}$ in time $O(\rank\sum_{v} (\dout{v})^3)$ $=O(\inputsize\rank\edegen^2)$.
\end{restatable}

\begin{algorithm}
    \caption{\starsalg$(\vec{G})$ \\Input: DAH $\vec{G} = (G,\pi)$ \\ Output: The number of stars $S$}\label{alg:stars}
\begin{algorithmic}[1]
\State $S \gets 0$
\For{$v \in V(G)$}
    \State $S \gets S + \binom{d^-_v}{3}+\binom{d^-_v}{2}d^+_v+\binom{d^+_v}{2}d^-_v$ \Comment{Case $(1)$}
    \For{$e_1,e_2,e_3 \in \Nout{u}$} 
        \If{$v = \source{e_1 \cap e_2 \cap e_3}$ \textbf{and} $\sink{e_1},\sink{e_2},\sink{e_3} \not\in e_1 \cap e_2 \cap e_3$} 
            \State $S \gets S + 1$ \Comment{Case $(2)$}
        \EndIf
    \EndFor
\EndFor
\State Return $S$
\end{algorithmic}
\end{algorithm}

\subsubsection{Weighted star counts}

Using similar ideas to the star counting procedure we can compute a more complex weighted sum.

Let an extended-star be a tuple $(\{e_1,e_2\},e_3)$ such that $\{e_1,e_2,e_3\}$ form a star and $e_1 \cap e_2 \setminus e_3 \neq \emptyset$.

We can show that this corresponds exactly to the patterns with one red region and at least one blue region, that is, patterns $1,3-8,10-16$. Moreover, for each pattern there will be as many valid tuples as the number of blue regions, giving us a weighted sum, that we will denote $\ws$.

\begin{definition} [$\ws$]
    Given a graph $G$, the weighted sum of extended stars $\ws$ is defined as the number of tuples $(\{e_1,e_2\},e_3)$ such that $e_1 \cap e_2 \setminus e_3 \neq \emptyset$ and $e_1 \cap e_2 \cap e_3 \neq \emptyset$.
\end{definition}

We can show the following:

\begin{restatable}{claim}{wsclaim} \label{clm:ws}
\begin{align*}
    \ws = &(\ct{1} + \ct{3} + \ct{6} + \ct{10}) \\+ &2(\ct{4} + \ct{5} + \ct{7} + \ct{8} + \ct{11} + \ct{12}) \\+ &3(\ct{13} + \ct{14} + \ct{15} + \ct{16}) 
\end{align*}
\end{restatable}

Let $e_1,e_2$ be the two edges such that $e_1 \cap e_2 \setminus e_3 \neq \emptyset$. Again we need to consider two main cases:
\begin{asparaenum}
    \item The edge $e_3$ has its endpoint in $e_1 \cap e_2$: for this to happen $e_1 \cap e_2$ must contain at least $2$ vertices. We iterate over each vertex $u \in V(G)$ and iterate over each pair of out-neighbors $e_1,e_2$ such that $u$ is the first vertex in their intersection. We can then obtain the count of tuples for this pair of edges by summing the in-degrees of every vertex in the intersection. We then subtract the number of hyperedges containing the intersection with its endpoint in $e_1 \cap e_2$.
    \item The edge $e_3$ does not have its endpoint in $e_1 \cap e_2$. We distinguish two sub-cases:
    \begin{asparaenum}
        \item $\{\sink{e_1}\} = \{\sink{e_2}\} = e_1\cap e_2\cap e_3$: We start at a vertex $u$ and look at pairs of out-neighbors like in the previous case. Then, if $e_1$ and $e_2$ end in the same vertex $v$ we look at the hyperedges starting from $v$.
        \item Otherwise, we can start at a vertex $u$ and look at triplets of out-edges from $u$. Then we manually check if the hyperedges intersect pairwise without the other hyperedge.
    \end{asparaenum}
\end{asparaenum}

\begin{figure}[htbp]
    \centering
    \includegraphics[width=.8\linewidth]{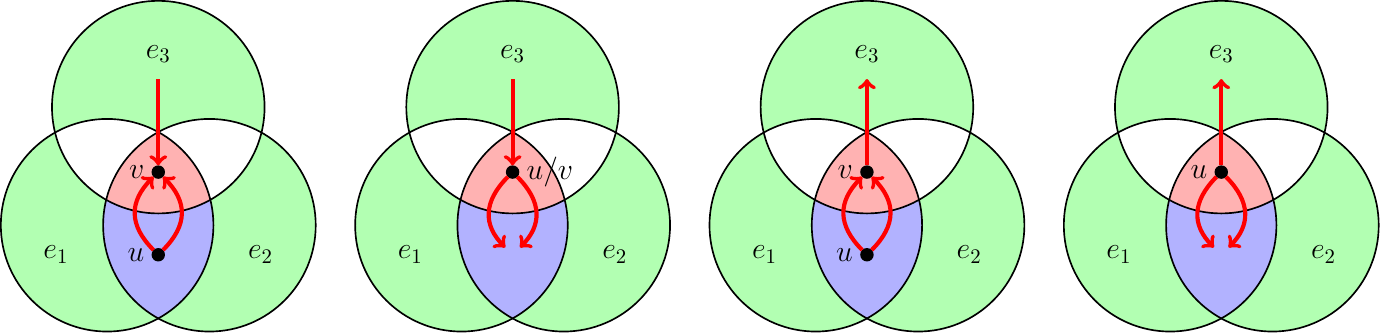}
    \caption{Different examples of how an extended star could appear in the directed hypergraph. 
    \\The two figures on the left correspond to case $(1)$, on the right $u$ is the common source of $e_1$ and $e_2$, but also the endpoint of $e_3$. In the example on left $v$ is a different vertex than $u$.
    \\The third figure is an example of case $(2.a)$, where $e_1$ and $e_2$ have a common endpoint $v$, which is a source of $e_3$. The right figure shows an example of case $(2.b)$, where $u$ is the common source of all the hyperedges.
    }
    \label{fig:extendedstars}
\end{figure}

\begin{algorithm}
\caption{\extstarsalg$(\vec{G})$ \\Input: DAH $\vec{G} = (G,\pi)$ \\ Output: The number of extended stars $\ws$}\label{alg:extended-stars}
\begin{algorithmic}[1]
\State $\ws \gets 0$
\For{$u \in V(G)$}
    \For{$e_1,e_2 \in \Nout{u}$}
        \If{$u = \source{e_1 \cap e_2}$ \textbf{and} $|e_1 \cap e_2| > 1$}
            \For{$v \in e_1 \cap e_2$}
                \State $\ws \gets \ws + \din{v}$ \Comment{Case $(1)$}
            \EndFor
            \For{$e_3 \in \Nout{u}\setminus\{e_1,e_2\}$}
                \If{$e_3 \supseteq e_1 \cap e_2$ and $\sink{e_3} \in e_1 \cap e_2$}
                    \State $WS \gets WS - 1$ \Comment{Subtract if $e_3 \supseteq e_1 \cap e_2$}
                \EndIf
            \EndFor
            \If{$\sink{e_1} = \sink{e_2}$}
                \State $v \gets \sink{e_1}$
                \For{$e_3 \in \Nout{v}$}
                    \If{$e_3 \cap e_2\cap e_1 = \{v\}$}
                        \State $\ws \gets \ws + 1$ \Comment{Case $(2.a)$}
                    \EndIf
                \EndFor
            \EndIf
        \EndIf
        \For{$e_3 \in \Nout{u}\setminus\{e_1,e_2\}$} 
            \If{$u = \source{e_1 \cap e_2 \cap e_3}$ \textbf{and} $e_1 \cap e_2 \setminus e_3 \neq \emptyset$ \textbf{and} $\sink{e_3} \not\in e_1 \cap e_2$}
                    \State $\ws \gets \ws + 1$ \Comment{Case $(2.b)$}
            \EndIf
        \EndFor
    \EndFor
\EndFor
\State Return $\ws$
\end{algorithmic}
\end{algorithm}
\begin{restatable}{lemma}{weightedstars}
    \Alg{extended-stars} computes $\ws$ in time 
    $O(\rank^2\sum_{v} d_v(\dout{v})^2)$ $=$ $O(\inputsize\rank^2 \edegen^2)$.
\end{restatable}

\section{Experiments} \label{sec:experiments}

We first discuss our experimental setup. 

{\bf Implementation.} 
All algorithms are implemented in C++, and all experiments are conducted on a standard laptop with a Ryzen 9 5900HS processor and $16$GB memory.
The code for our algorithm is given in \cite{ditch}.
We used 9GB as maximum memory budget for all experiments.

{\bf Baselines.} 
In the experiments, we use \baseline~\cite{LeJiSh20}, \baselineadv~\cite{YiWaZh+24}, Hypergraphx~\cite{hypergraphx/compnet/LotitoBNMB26}, and DuckDB as baselines to compare our package, \ditch.

Recall that there are a total of 26 patterns that we want to count, out of which the first 20 are closed patterns and the last 6 are open patterns. \baseline\ computes the count for both closed and open patterns while \baselineadv, Hypergraphx, and DuckDB compute the counts for open patterns.

DuckDB and Hypergraphx do not provide native support for the edge-based hypertriangle counts. Hypergraphx supports vertex-based notions of hypergraph motifs. 
We used Hypergraphx’s internal line-graph implementation together with Networkx’s triangle-listing implementation. To classify the triangles, we used Python set operations on the three corresponding hyperedges to determine their hypertriangle types.

For DuckDB, we represented each hypergraph as a two-column incidence table, with one row (e,v) for each vertex v incident to hyperedge e. From this table, we computed pairs of intersecting hyperedges together with their intersection sizes, effectively constructing a weighted line-graph representation, and used this relation to compute the hypertriangle counts. DuckDB was used in in-memory mode with 10 threads. All queries were structured to use pipelining and minimize unnecessary materialization.

{\bf Datasets.} In our experiments, we use hypergraph datasets obtained from \cite{datasets,BeAbSc+18}.
The dataset details are given in \Tab{summary-datasets}. We preprocess the dataset to remove timestamps and duplicates.
These datasets are of five different types, as indicated in the dataset name.

\begin{asparaitem}
    \item[A:] Co-author: Vertices are researchers, and hyperedges correspond to coauthors on papers. Includes coauth-DBLP (A-DBLP), coauth-MAG-Geology (A-geology) and coauth-MAG-History (A-History) datasets \cite{BeAbSc+18,SiShSo15}.
    \item[C:] Social contact: 
Vertices represent individuals, and hyperedges represent contact groups or interactions among them. Includes contact-high-school (C-hschool) and contact-primary-school (C-pschool) datasets \cite{BeAbSc+18, MaFoBa15, StVoBa11}.
    \item[E:] Email: 
Vertices correspond to email accounts, and hyperedges represent emails and contain senders and recipients. Includes email-Eu (E-eu) and email-Enron (E-enron) datasets \cite{BeAbSc+18, YiBeLe17, LeKlFa07}.
    \item[Th:] Threads: 
Vertices represent users, and hyperedges discussion threads in which users participate. Includes threads-math-sx (Th-math) and threads-ask-ubuntu (Th-ubuntu) datasets \cite{BeAbSc+18}.
    \item[Tg:] Tags:
Vertices represent tags, and hyperedges represent posts containing those tags. Includes the tags-ask-ubuntu (Tg-ubuntu) dataset \cite{BeAbSc+18}.
\end{asparaitem}

We note that largest instance A-dblp has more than two million hyperedges, and 
three datasets have ranks in the hundreds. The Tg-ubuntu dataset is extremely dense,
with a largest degree of 12K and 13k. A linear dependence on the rank is expected
for any algorithm, and large hyperedges lead to significant increases in runtime.
Overall, these numbers show why hypergraph motif counting
is significantly more challenging than the graph setting.

\begin{table}[ht]
\centering
\begin{tabular}{lrrrrrr}
\toprule

Dataset & $|V|$ & $|E|$ & $r$ & $|e|_{avg}$ & $\max_v d_{v}$ & $\edegen$ \\

\midrule
E-enron & 149 & 1,514 & 37 & 3.05 & 117 & 52 \\

C-hschool & 328 & 7,818 & 5 & 2.33 & 148 & 46 \\

C-pschool & 243 & 12,704  & 5 & 2.42 & 261 & 98 \\

E-eu & 1,006 & 25,148 & 40 & 3.55 & 917 & 306 \\

Tg-ubuntu & 3,030 & 147,222	& 5	& 3.39 &	12,930 & 2617 \\

Th-ubuntu & 200,975 & 166,999 & 14 & 1.91 & 2,170 & 115 \\

Th-math & 201,864 & 595,749 & 21 & 2.45 & 11,358 & 617 \\

A-history & 1,034,877 & 896,062 & 925 & 1.57 & 402 & 294 \\

A-geology & 1,261,130 & 1,204,704 & 284 & 3.17 & 716 & 174 \\

A-dblp & 1,930,379 & 2,467,389 & 280 & 3.14 & 846 & 221 \\
\bottomrule
\end{tabular}
\caption{Summary of 10 Datasets (sorted by $|E|$ ascending)}
\label{tab:summary-datasets}
\end{table}

{\bf Implementation Details.}
We use a double CSR to represent the hypergraph: one that for each vertex shows the hyperedges that contain it, and another one that for each hyperedge show the vertices contained on it.

The implementation follow a similar structure to the algorithms in \Sec{degen} and \Sec{hypertri}, but there are small differences in the implementation with the aim to improve the performance. We also store explicitly the entire set of ancestors and descendants of each hyperedge, as in practice they take limited space without affecting performance (therefore the theoretical $O(h)$ space bound does not apply to the implementation).

\subsection{Main Results}
\label{sec:main_results}

\begin{figure}
	\centering
	\includegraphics[width=0.75\linewidth]{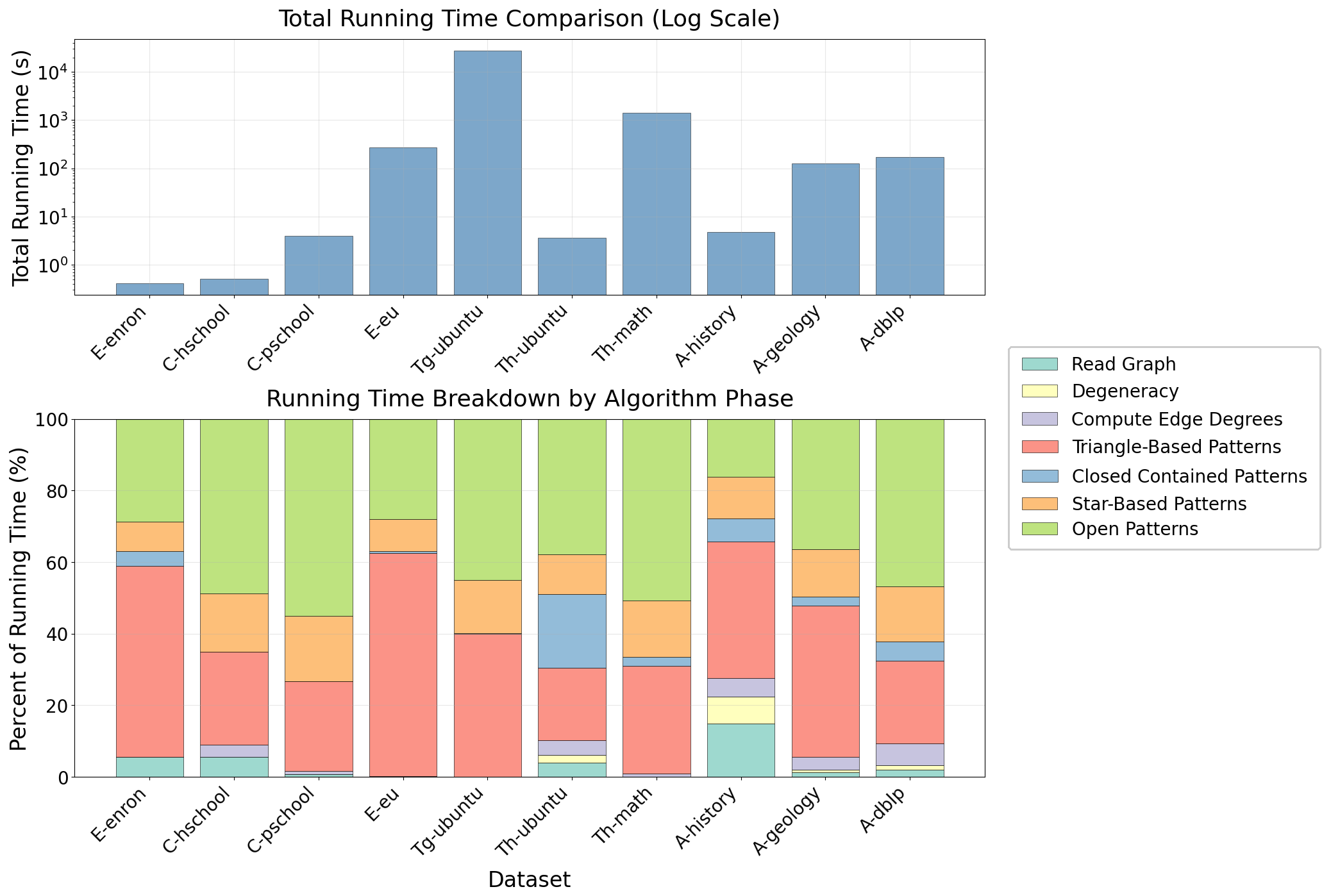}
	\caption{Breakdown of runtime of \ditch}
	\label{fig:runtime-breakdown}
\end{figure}

\textbf{Runtime Comparison.} \Fig{runtime} compares the running time of \ditch\ with \baseline, \baselineadv, Hypergraphx, and DuckDB. Algorithms that exceed the 9GB memory limit are not shown. As seen in \Fig{runtime}, \baselineadv\ is always faster than \baseline\ and Hypergraphx (the y-axis is in logscale).
\ditch\ is faster than the state-of-the-art \baselineadv\ and DuckDB, up to even 100x, on nearly all datasets. The only exception is with \baselineadv\ on E-enron where the difference is less than few seconds.
For example, in the Th-ubuntu dataset, \ditch\ runs in less than 10 seconds,
while \baselineadv\ takes hours.
Note that \baseline\ and \ditch\ compute both open and closed patterns while all others only compute closed.

We observe that \ditch\ finishes in under 30 minutes for {9 out of the 10}.
The only exception 
is {Tg-ubuntu}, a graph for which no previous exact counting algorithm had 
ever finished. With \ditch, we are able to compute the exact counts in about 
{8 hours}, obtaining exact counts for the first time.

Interestingly, DuckDB is highly competitive with the specialized \baselineadv\ algorithm: it is slightly faster on C-pschool and very close on the other datasets where both complete. This likely reflects both DuckDB’s optimized execution and the fact that the two approaches rely on similar intermediate representation.
Hypergraphx also completes on more datasets than \baseline, despite being a general hypergraph analytics library.

\textbf{Memory Comparison.}
\Fig{memory} compares the memory usage of the three algorithms. 
\baseline\ exceeds the 9GB limit on 6 out of 10 datasets, Hypergraphx exceeds on 5, while \baselineadv\ and DuckDB exceed the limit on 3 datasets.
In contrast, \ditch\ stays well within the memory limit, and requires less than 800MB for all datasets.

\baseline\ uses the most memory on all instances except E-enron. DuckDB is competitive in memory, often matching or outperforming \baselineadv. However, all baselines rely on storing the pairwise hyperedge-intersection structure. 

\ditch\ avoids this bottleneck, achieving 6\text{--}100x memory savings. On Tg-ubuntu, \ditch\ uses only 58.95MB, while prior exact methods fail even with memory limits well above 9GB. It stays below 370MB on every dataset except DBLP.

\textbf{Curtailing the Tail.} In \Fig{deg-dis}, we compare (in log-log plots) the degree distribution
with the outdegree distribution for four datasets. Specifically, the x-axis is the (out)degree,
and the y-axis is the number of vertices with that degree. We see the power of the orientation.
The frequency of high outdegree vertices is significantly smaller than that of high degree
vertices. Since neighborhood size is the main bottleneck for efficiency, this demonstrates
how the orientation effectively cuts these sizes down. In \Tab{summary-datasets},
we see a comparison of the maximum degree vs the T-degeneracy. While the difference is not
as dramatic as the tails, we consistently see a reduction by a factor of 2-5.
Since all hypertriangle algorithms pays at least quadratic in the neighborhood size,
these reductions significantly lower the running time.

\subsection{Runtime Analysis}

{\bf Runtime breakdown.}
\Fig{runtime-breakdown} summarizes the runtime of \ditch\ across all datasets. The top panel reports the total runtime, while the bottom gives the percentage contribution of each phase of \ditch. 

The degeneracy and compute edge degrees phases take substantially less time than the remaining phases. So the preprocessing computations of \Sec{degen} are negligible, consistent
with the linear time bounds obtained in Lemmas~\ref{lem:mb} and \ref{lem:deg} and Corollary~\ref{cor:hyp-deg}.
Typically, the {triangle-based patterns} take a large fraction of the time,
and the various algorithms (containment, star) of \Sec{contain} and \Sec{star} are quicker.
We note that open pattern counting is also expensive, and takes anywhere from
a third to half the time. Note that \ditch\ is compared with the best methods
that do \emph{not} compute open patterns, showcasing the significant benefits 
of the orientation techniques.

\begin{figure}
\centering
\includegraphics[width=0.6\linewidth]{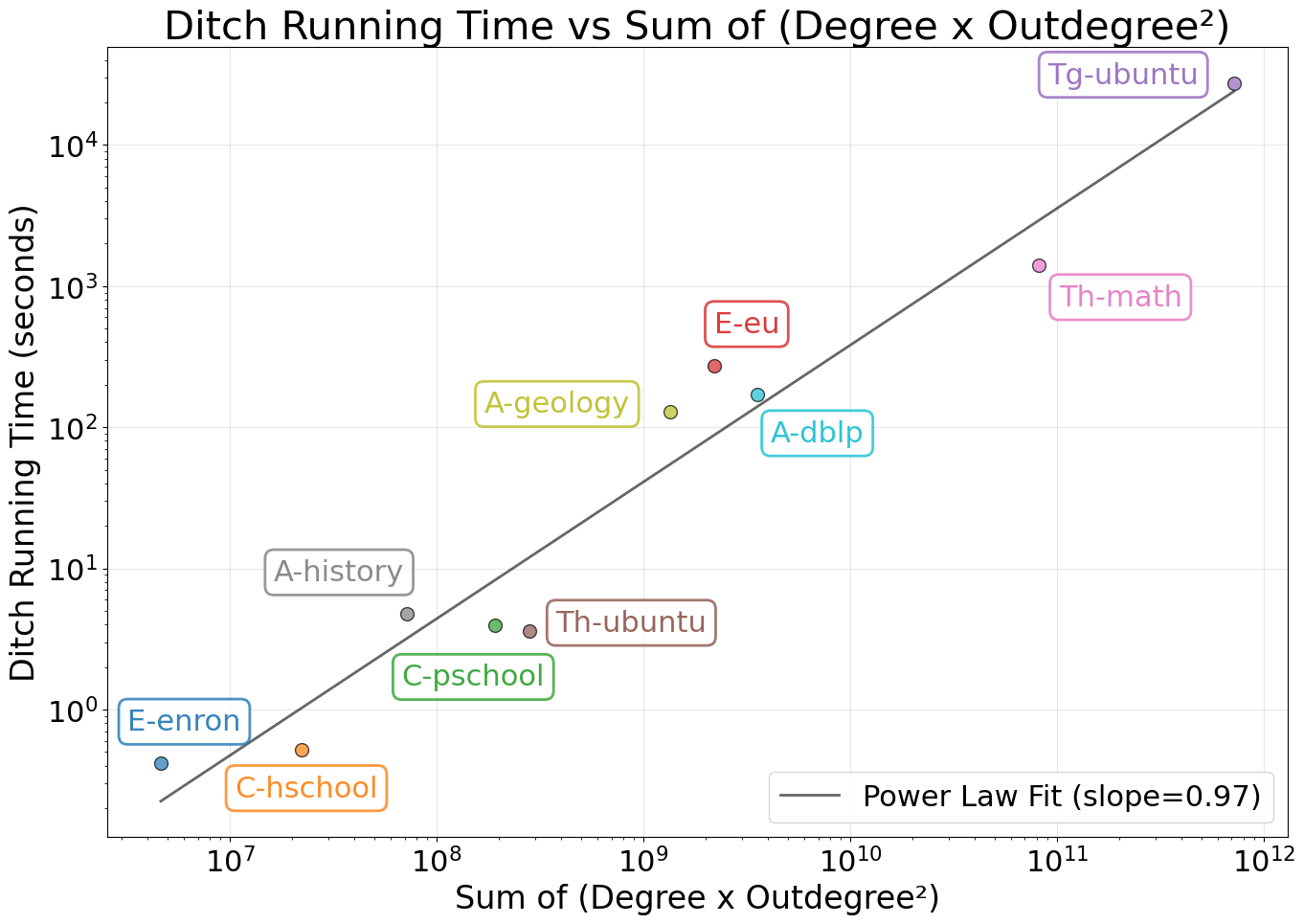}
\caption{Correlation of runtime and $O(\sum_v d_v (d^+_v)^2)$}
\label{fig:correlation}
\end{figure}

{\bf Validating the mathematical running time.} The total theoretical running time of \ditch\
is $O(\sum_v d_v (d^+_v)^2)$. In \Fig{correlation}, we plot this quantity versus the wall
clock running time.
We see a near perfect correlation, showing that our mathematical analysis accurately captures
the actual running time.

\subsection{ Memory Usage Analysis}
\label{sec:memory}
\begin{figure}
\centering
\includegraphics[width=0.95\linewidth]{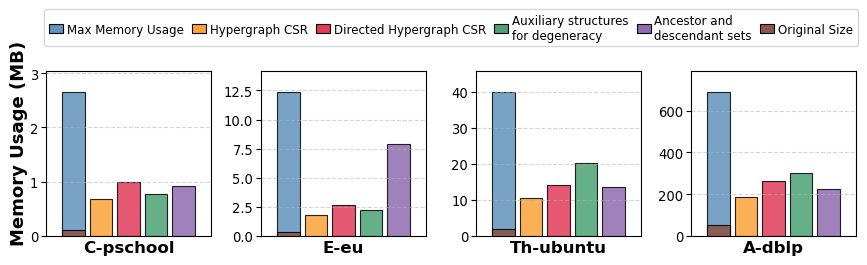}
\caption{Detailed memory usage across different datasets.}
\label{fig:memory_detailed}
\end{figure}

In \Fig{memory_detailed} we can see a breakdown of the memory usage for a subset of the datasets. Only the memory for the Hypergraph and Directed Hypergraph CSRs are persistent, while the auxiliary structures for the degeneracy and the ancestor and descendant sets are only in use during the corresponding algorithms. Only in the case of E-eu, the memory for storing the ancestor and descendant sets surpasses the memory for storing the hypergraphs. Other data structures used, like the in-degree arrays, are small in comparison and barely affect the maximum memory usage.

In \Fig{memory_slope}, we see an excellent power law
fit between the input and peak memory usage. Although our theoretical bound was $O(\inputsize)$, in practice the relation is sublinear.

\subsection{Pattern Count and \ditch}

Table~\ref{tab:pattern-counts} reports the total number of closed and open patterns for each dataset, along with the percentage contribution of the most frequent closed pattern (P9) and the most frequent open pattern (P26). From Table~\ref{tab:pattern-counts}, we observe that the closed and open pattern groups are comparable in magnitude across all datasets, but in both groups the \emph{distribution is highly skewed}: a small number of patterns account for the vast majority of occurrences. P9 alone contributes over 30\% of all closed patterns in every dataset, and exceeds 70\% in 7 out of 10 datasets.
A similar trend appears in the open group, where P26 overwhelmingly dominates, with over 55\% contribution in all datasets and over 80\% in 7 out of 10 datasets. 

\begin{figure}
	\centering
	\includegraphics[width=0.6\linewidth]{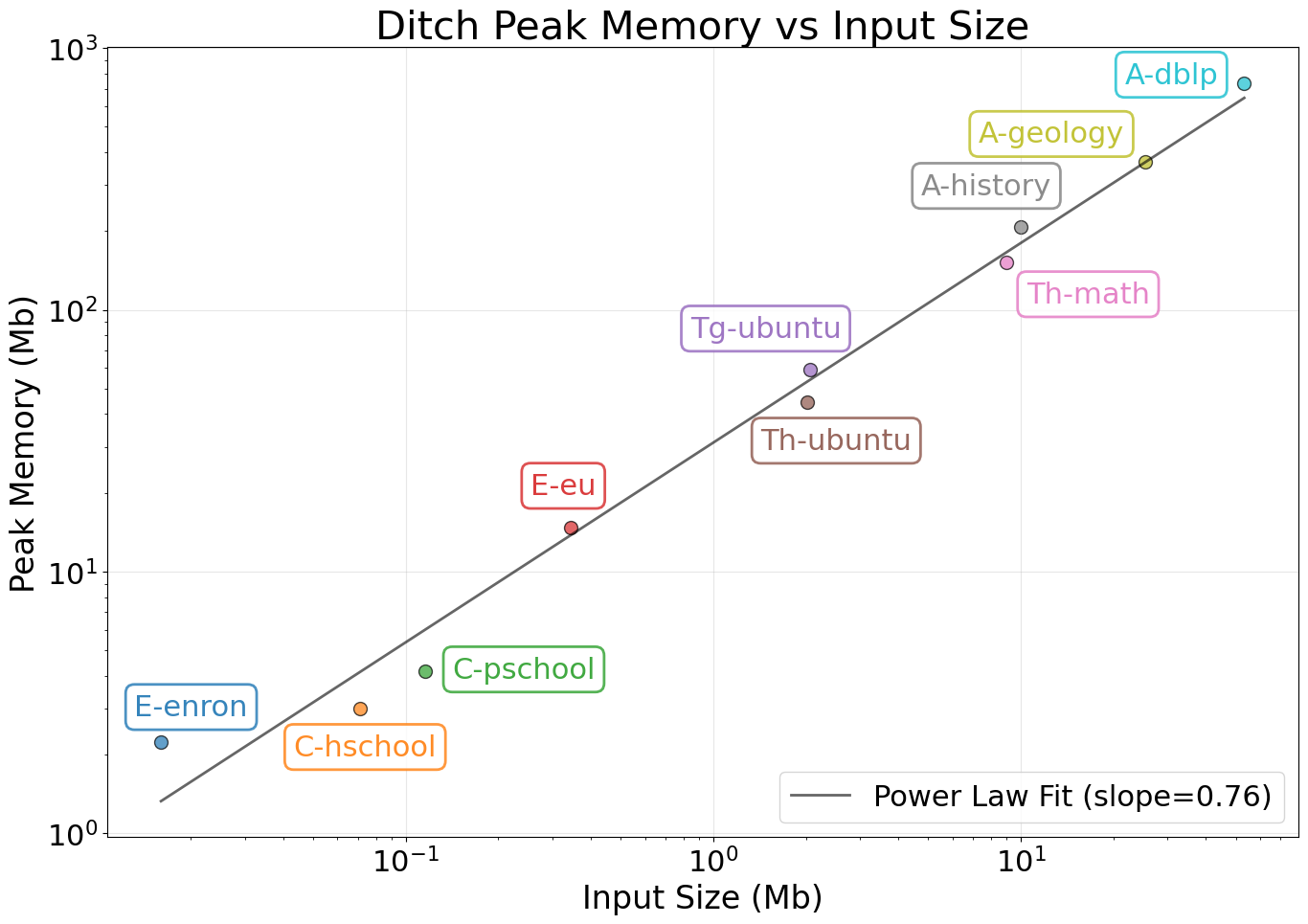}
	\caption{Correlation of peak memory usage and input size.}
	\label{fig:memory_slope}
\end{figure}

This heavy concentration of counts in a small number of patterns aligns with our algorithmic design: as discussed earlier, our method counts many high-frequency patterns using formulas instead of enumerating them explicitly, resulting in substantial runtime gains.

A particularly striking example is the previously intractable Tg-ubuntu dataset. Its high density leads to an enormous total number of patterns; on the order of \emph{trillions}. This massive pattern count makes enumeration based approaches infeasible. 

\begin{table}[ht]
\centering
\begin{tabular}{lrrrrr}
\toprule
Dataset & Closed Count & P9\% & Open Count & P26\% \\
\midrule
E-enron & 2,509,330 & 30.8 & 7,696,592 & 69.5 \\
C-hschool & 16,580,005 & 77.4 & 53,178,289 & 56.6 \\
C-pschool & 133,253,436 & 82.6 & 484,270,407 & 64.3 \\
E-eu & 1,434,783,682 & 44.8 & 6,409,149,896 & 86.5 \\
Tg-ubuntu & 1,288,391,437,610 & 78.2 & 3,051,640,275,979 & 98.0 \\
Th-ubuntu & 6,919,513,919 & 98.7 & 4,545,001,044 & 91.7 \\
Th-math & 1,069,446,477,335 & 98.4 & 1,152,606,416,054 & 97.4 \\
A-history & 50,000,497 & 81.8 & 39,815,923 & 83.5 \\
A-geology & 1,384,373,331 & 64.2 & 5,153,483,562 & 93.7 \\
A-dblp & 8,354,454,505 & 83.7 & 18,285,470,142 & 93.5 \\
\bottomrule
\end{tabular}
\caption{Counts of closed and open patterns by dataset, along with \% of most frequent closed (P9) and open (P26) pattern.}
\label{tab:pattern-counts}
\end{table}

\subsection{Data analysis \&  fine-grained counts} \label{sec:fine-grained}

Our algorithm allows for a fine-grained analysis of hypertriangle patterns.
As an example, we can compute the distribution of intersection sizes (the red region in \Fig{types})
over the different instances of pattern $9$ - discussed briefly in Example~1 in the introduction and expanded here. \ditch{} is able to compute these with no additional overhead, since 
it explicitly ``touches'' all instances with at least $2$ vertices in the intersection.

\Fig{star_distribution} shows the distribution of stars for A-dblp and A-geology depending on the size of the intersection. In these co-authorship networks, the intersection size represents how many authors in common the three papers forming a star have.

As we can see, most instances of the pattern $9$ contain a single vertex in the intersection, that is, there is a single author in common between the papers.

Moreover, the frequency decreases exponentially with the size of the intersection. The instances with multiple vertices in the intersection are of special interest, as in this case they represent large groups of authors that work together with distinct sets of co-authors. For A-dblp we have that there are exactly $31$ stars with at least $10$ authors in the intersection, while for A-geology there are $183$. Furthermore, we can recover $5$ stars in A-dblp where there is a set of $25$ authors in the intersection of the three papers.

Our algorithm can be easily modified to output all the instances with intersection size greater than $1$ in no additional time. 

\begin{figure}
	\centering
		\includegraphics[width=0.85\linewidth]{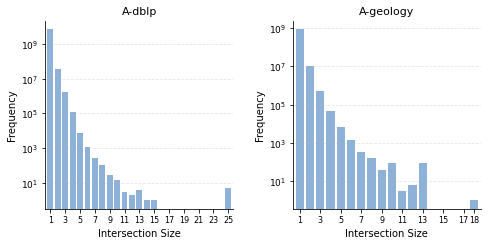}
    \caption{\small Distribution of pattern $9$ according to the size of the intersection (red region).}
	\label{fig:star_distribution}
\end{figure}

\section{Conclusion}

We present \ditch, an efficient algorithm for counting hypertriangles that leverages a novel theoretical framework generalizing degeneracy and orientations to hypergraphs. Through experiments, we show that orientations with small out-degrees enable \ditch\ to achieve up to 100x speedup and use significantly lower (6-100x) memory compared to existing methods. We believe that degeneracy- and orientation-based techniques will play a fundamental role in advancing hypergraph motif counting just as they have in graph algorithms. A concrete next step would be to extend DITCH to count motifs with $4$ or more hyperedges, handle non-binary or count-based intersection conditions on the 7 regions (as in~\cite{vldb/LeeYKKS24,bigdataconf/NiuAAS24}) and to support vertex-based hypertriangle motifs.

\section*{Acknowledgments}
Daniel Paul-Pena, Vaishali Surianarayanan and C. Seshadhri are supported by NSF DMS-2023495, CCF-1740850, 2402572. Deeparnab Chakrabarty is Supported by NSF awards CCF-2041920, 2402571. Vaishali Surianarayanan is also supported by the UCSC Chancellor’s Postdoctoral Fellowship Yale Toni Massini Postdoctoral Fellowship.

\bibliographystyle{alphaurl}
\bibliography{hypergraphs,subgraph_counting_doi,motifcount,hypertriangles}

\clearpage

\appendix

\section{Missing proofs} \label{sec:proofs}

\subsection*{Proofs for \Sec{degen}}

\matulabeck*
\begin{proof}
	We first show that Algorithm \ref{alg:degeneracy} returns an ordering with maximum outdegree $\leq \edegen(G)$.
	With Lemma~\ref{lem:ub} this would show equality. After this, 
	we will argue that the runtime of the algorithm is $O(n+mr)$ when using the right data structures.
	
	Let $\pi = (v_1,\dots,v_n)$ be the ordering of vertices output by Algorithm \ref{alg:degeneracy}. Let $G_i$ be the hypergraph before the removal of the $i$-th vertex. We first show by induction that $G_i$ is equal to the \Multihyper{} of $G$ induced by $(v_i,\dots,v_n)$.
	
	The base case is trivial as $G_0 = G = G\langle V(G)\rangle$. For the inductive step, assume $G_i = G\langle \{v_i,\dots,v_n\} \rangle$, we show that $G_{i+1} = G\langle \{v_{i+1},\dots,v_n\} \rangle$:
	\begin{asparaitem}
		\item  For the vertex set, we have $V(G_{i+1}) $ $= V(G_i)\setminus\{v_i\}$ \\$= V( G\langle \{v_{i},\dots,v_n\}\rangle) \setminus\{v_{i}\} $ $=  V( G\langle \{v_{i+1},\dots,v_n\}\rangle)$.
		\item For the hyperedge set, first note that every hyperedge in $G\langle \{v_{i+1},\dots,v_n\}\rangle$ corresponds to some edge in $G\langle \{v_{i},\dots,v_n\}\rangle$ (and therefore in $G_i$. Let $e$ be a hyperedge in $V(G_i)$, and $e' = e\setminus\{v_i\}$, there are two possible cases:
		\begin{asparaitem}
			\item $|e'|=1$: In this case $e'$ will be removed from $E(G_{i+1})$ and will not be present in $G\langle \{v_{i+1},\dots,v_n\}\rangle$, as hyperedges of arity less than $2$ are removed in both cases.
			\item $|e'|>1$: In this case we will have that $e'$ is in both $E(G_{i+1})$ and $G\langle \{v_{i+1},\dots,v_n\}\rangle$.
		\end{asparaitem}
	\end{asparaitem}
	Now, we show that the outdegree of $v_i$ in $\vec{G}=(G,\pi)$ is equal to its degree in $G\langle \{v_i,\dots,v_n\}\rangle$. Note that a hyperedge $e$ contributes to $\dout{v_i}$ only if $v_i \in e$ and $v_i \neq \sink{e}$, this means that there is a vertex $v' \in e$ with a later ordering than $v_i$ in $\pi$. This implies that, at the removal time of $v_i$, $|e \cap\{v_i,\dots,v_n\}|\geq 2$, and therefore the hyperedge is present in $G\langle \{v_i,\dots,v_n\}\rangle$. Conversely, every hyperedge in $G\langle \{v_i,\dots,v_n\}\rangle$ must have $|e \cap\{v_i,\dots,v_n\}|\geq 2$, which implies that there will be a later vertex than $v$ in $e$ according to $\pi$.
	
	Putting everything together, we get
	\begin{equation*}
		\Delta^+(\vec{G}) =\max_{v\in V(G)} \dout{v} = \max_{i=1}^n(d_{v_i} \text{ in }G\langle \{v_i,\dots,v_n\}\rangle) \leq \edegen(G)
	\end{equation*}
	where the final inequality follows since by choice $v_i$ is the minimum degree vertex in $\trimmed{G}{\{v_i, \ldots, v_n\}}$ and 
	since $\edegen(G)$ is the maximum over all subsets of $V$.

	Now, we analyze the runtime of the algorithm. One can construct a degree structure similar to \cite{MaBe83} in time $O(h)$, that allows to find the vertex with lowest degree in time $O(d_{v_i})$. For the removal we only need to look at every hyperedge adjacent to $v_i$. When deleting a hyperedge we will need to iterate over its vertices and update their degrees, which will take constant time per vertex. Note that every hyperedge and every vertex will be removed a single time. Therefore, the complexity of the algorithm will be:
	
	\begin{equation*}
		O(h) + \sum_{v \in V(G)} O(d_{v}) + \sum_{e\in E(G)} O(|e|) = O(h)\qedhere
	\end{equation*}
	
\end{proof}
\subsection*{Proofs for \Sec{triangle}}

First we formally introduce the notion of a witness:

\begin{definition} [Witness of $\{e_1,e_2,e_3\}$]\label{def:witness}
    Given a triplet of hyperedges $\{e_1,e_2,e_3\}$, the witness $w(e_1,e_2,e_3)$ is a sorted triplet of vertices formed by:
    \begin{asparaitem}
        \item If $\type(e_1,e_2,e_3) \in [13,20]$: $\source{e_1\cap e_2\setminus e_3}$, $\source{e_1\cap e_3\setminus e_2}$ and $\source{e_2\cap e_3\setminus e_1}$.
        \item If $\type(e_1,e_2,e_3) \in \{4,5,7,8,11,12\}$: Assume without loss of generality that $e_1\cap e_2\setminus e_3=\emptyset$,  $w(e_1,e_2,e_3)$ will be formed by $\source{e_1\cap e_2\cap e_3}$, $\source{e_1\cap e_3\setminus e_2}$ and $\source{e_2\cap e_3\setminus e_1}$.
    \end{asparaitem}
\end{definition}

\triangles*
\begin{proof}
    Lines $5$-$9$ can be computed in time $O(\rank)$. Therefore the runtime of the algorithm will be:
    \begin{align*}
        & \sum_{u} \sum_{e_1 \in \Nout{u}} \sum_{v \in e_1}  \dout{u}\dout{v}O(\rank) 
        = O\left(\rank \sum_e\sum_{u,v \in e}\dout{u}\dout{v}\right) 
        \\= & O\left(\rank \sum_{e} \sum_{u,v \in e}(\dout{u})^2 +(\dout{v})^2\right) 
         =  O\left( \rank \sum_e |e| \sum_{v \in e} (\dout{v})^2 \right)
         \\= & O\left( \rank \sum_v (\dout{v})^2 \sum_{e \ni u }|e|\right) = O\left( \rank^2 \sum_v d_v (\dout{v})^2\right)
     \end{align*}

     For the correctness we can see that the counts only increase when we are in an iteration such that $u,v$ are the first two elements of the witness and $w$ is in the intersection of $e_2,e_3$. First we show that each triplet $e_1,e_2,e_3$ will be counted only once.

     Assume otherwise, we have two (or more) instances of a triplet being counted. For that to happen the same hyperedges must appear in different loops, switching positions. Let $e_1,e_2,e_3$ be assigned in their corresponding loops in the first instance (and therefore $u \in e_1\cap e_2$, $v \in e_1\cap e_3$ and $w \in e_2 \cap e_3$) and consider the possible switches that can happen. We can see that it will always lead to at least two vertices being in the intersection of the three hyperedges, and therefore $u,v,w$ will not be a valid witness:
     \begin{asparaitem}
         \item $e_1 \leftrightarrow e_2$: In this case we will have $v \in e_2$ and $w \in e_1$.
         \item $e_1 \leftrightarrow e_3$: In this case we will have $u \in e_3$ and $w \in e_1$.
         \item $e_2 \leftrightarrow e_3$: In this case we will have $u \in e_3$ and $v \in e_2$.
         \item $e_1\rightarrow e_2 \rightarrow e_3 \rightarrow e_1$ or $e_1\leftarrow e_2 \leftarrow e_3 \leftarrow e_1$: In this case the three vertices will be in $e_1\cap e_2 \cap e_3$.
     \end{asparaitem}

    Only left is to show that the triplet will be counted once. Two of the hyperedges must be out-neighbors of $u$, as they contain $u$ and either $v$ or $w$ (or both), with at least one of them containing $v$. Moreover, the remaining hyperedges must be an out-neighbor of $v$, as it must contain both $v$ and $w$. Therefore the algorithm will find the witness the triplet $e_1,e_2,e_3$ when using $u$ and $v$ in the first and third loops.\qedhere
\end{proof}

\subsection*{Proofs for \Sec{contain}}

\typeone*
\begin{proof}
    Pattern $1$ is formed by three hyperedges $e_1,e_2,e_3$ such that $e_1$ is a descendant of $e_2$ and $e_3$ is an ancestor of $e_2$. For a fix hyperedge $e$, we can get the number of instances where it act as the middle hyperedge of a type $1$ pattern by multiplying the number of descendants with the number of ancestors. Repeating over all hyperedges gives the total number of copies of Pattern $1$.
\end{proof}

\typetwothree*
\begin{proof}
    Patterns $2$ and $3$ are formed by a hyperedge $e_1$ that it is contained by two hyperedges $e_2$ and $e_3$. We can iterate over each pair of ancestors of each hyperedge to get the total count of each such patterns.
    
    The running time follows from the observation that if $e' \supset e$ (and $e$ has arity\footnote{The pairs where the smaller hyperedge has arity $1$ can be handled separately, as there will be at most $O(m r)$ such pairs.} at least $2$) then at least one vertex in $e$ has both $e$ and $e'$ as outedges. Hence, $\Nparents{e}^2 \leq \sum_{v \in e} (d^+_v)^2$.
    Using the same idea from \Lem{deg}, $\sum_e \Nparents{e}^2 \leq \sum_e \sum_{v \in e} (d^+_v)^2$ $\leq \sum_v d_v (d^+_v)^2$. Running Line 3 takes $O(r)$ time.
\end{proof}

\typesix*
\begin{proof}
    The term on the right gives for every hyperedge the product of the number of parents and number of hyperedges intersecting it. Let $e_1$ be a hyperedge, let $e_2$ be a hyperedge intersecting it and $e_3$ a hyperedge containing $e_1$. We distinguish two cases:
    \begin{asparaitem}
        \item $e_3$ is a parent of $e_2$: In this case, we have that the triplet $e_1,e_2,e_3$ will be of type $4$ or $5$. Moreover we will count this triplet twice, as $e_2$ intersects with $e_1$ and has $e_3$ as a parent.
        \item $e_3$ is not a parent of $e_2$: In this case the triplet $e_1,e_2,e_3$ must of type $6$,$7$ or $8$.\qedhere
    \end{asparaitem}
\end{proof}

\subsection*{Proofs for \Sec{star}}

\stars*
\begin{proof}
	The running time follows since we iterate over all vertices and all triples of out-neighbors.
    The computation on line $3$ can be done in constant time (one can precompute the in-degrees in $O(m\rank)$ time) and the check on line $5$ can be done with a linear search in time $O(r)$.

    Now we prove the correctness. If a triplet $e_1,e_2,e_3$ do not share a vertex, then they will not contribute to the count $S$, as there is no vertex that has all of them as out-neighbors or in-neighbors. Now, let $e_1,e_2,e_3$ be a triplet with non-zero intersection. And let $v = \source{e_1\cap e_2 \cap e_3}$. We distinguish the two cases, as shown in \Fig{stars}.
    \begin{asparaenum}
    \item One of the hyperedges $e_1$,$e_2$,$e_3$ has its endpoint in $e_1 \cap e_2 \cap e_3$. Note that if more than one hyperedge have the endpoint in $e_1 \cap e_2 \cap e_3$, then they must have the same endpoint. Let $v$ be that vertex. We have three sub-cases: either $1$ hyperedge ends in $v$, two hyperedges end in $v$, or all the hyperedges end in $v$. The contribution of line $3$ over all vertices is: $\sum_{v\in V(G)} \frac{1}{6}(\din{v})^3 + \frac{1}{2}(\din{v})^2\dout{v} + \frac{1}{2}\din{v} (\dout{v})^2$.
    
    The first term is equal to the number of triplets that have their endpoint in $v$, the second term is the number of triplets where $2$ vertices have their endpoint in $v$ and the other is an out-neighbor of $v$ and the third term is the number of triplets where $1$ vertex has the endpoint in $v$ but the other $2$ does not.
    
    \item $e_1$,$e_2$,$e_3$ do not have their endpoint in $e_1 \cap e_2 \cap e_3$. Because $e_1 \cap e_2 \cap e_3 \neq \emptyset$ there will be at least a vertex $v$ that have all these hyperedges as out-neighbors. The algorithm loops over all triplets of out-neighbors of each vertex, so it will check all possible copies of this type.

    Note that also a triplet gets counted only when looking at the vertex $v$ which is first in the intersection. Therefore we will count each valid triplet only once.\qedhere
\end{asparaenum}
    
\end{proof}

\wsclaim*
\begin{proof}
    Let $(\{e_1,e_2\},e_3)$ be a tuple such that $e_1\cap e_2 \setminus e_3 \neq \emptyset$ and $e_1\cap e_2 \cap e_3  \neq \emptyset$. Note that every instance of patterns $1,3,6,10$ contains one such tuple, patterns $4,5,7,8,11,12$ contains two, and patterns $13$ to $16$ contains three such tuples. Any other pattern either has $e_1\cap e_2 \cap e_3 = \emptyset$ or $e_1\cap e_2 \setminus e_3 = \emptyset$ for all their tuples $(\{e_1,e_2\},e_3)$.
\end{proof}

\weightedstars*
\begin{proof}
    For case $1$, the code loops over every pair of out-neighbors of every vertex and then every through every vertex in their intersection, so it will take $O(\rank \sum_{v} (\dout{v})^2)$. For cases $2.b$ and subtracting the extra counts the algorithm loops over triplets of out-neighbors and performs checks that can be computed in $O(\rank)$ time, therefore it will take $O(\rank \sum_{v} (\dout{v})^3)$. 
    
    Finally, for case $2.a$, the algorithm iterates over pairs of out-neighbors and then through out-neighbors of their sink (if it is the same). Thus we will have:
    \begin{align*}
        &O\left(\sum_{v} \sum_{e_1,e_2\in \Nout{v}} \dout{\sink{e_1}}\rank\right) 
        = O\left(\rank \sum_v \sum_{e_1\in \Nout{v}}\dout{v}\dout{\sink{e_1}}\right) 
        \\ = &O\left(\rank\sum_v\sum_{e_1 \in \Nout{v}} (\dout{v})^2 + (\dout{\sink{e_1}})^2\right)
        = O\left(\rank\sum_{v}(\dout{v})^3\right) + O\left(\rank^2\sum_{e_1} (\dout{\sink{e_1}})^2\right)
        \\ = & O\left(\rank\sum_{v}(\dout{v})^3\right) + O\left(\rank^2\sum_{v} \din{v}(\dout{v})^2\right) = O\left(\rank^2\sum_{v} d_v(\dout{v})^2\right)
    \end{align*}
    
    Now we prove the correctness. We wish to show that we will count every tuple $(\{e_1,e_2\},e_3)$ with $e_1\cap e_2 \setminus e_3 \neq \emptyset$ and $e_1\cap e_2 \cap e_3 \neq \emptyset$. First, we can check that if a triplet does not satisfy these conditions it will not affect the output of the algorithm.
    \begin{asparaenum}
        \item If $e_1\cap e_2 \cap e_3 = \emptyset$, then it is not possible that $e_3$ has its sink in $e_1\cap e_2$, so it will not contribute to line $6$. It also can not appear as an out-neighbor of a vertex in $e_1\cap e_2$ or of $u$, and therefore it can not appear as the hyperedge in loops in lines $7,12,15$.
        \item $e_1\cap e_2 \setminus e_3 = \emptyset$, then it can not be counted in line $17$ as the if of line $16$ explicitly checks this condition. If $e_3$ has its endpoint in $e_1\cap e_2$ then it will be counted in line $6$, however it will also be subtracted in line $9$ because we will have $e_3 \supseteq e_1\cap e_2$. Finally it is not possible to reach line $14$ in this case, because line $13$ can only be true if $|e_1 \cap e_2\cap e_1| = 1$ and line $4$ requires $|e_1\cap e_2|\geq 2$, both are not possible if $e_1\cap e_2 \setminus e_3 = \emptyset$.
    \end{asparaenum}
    Now, we show that if a tuple satisfies both conditions then it will be accounted for in the algorithm, note that $|e_1\cap e_2|>1$ in this case. We follow the same case analysis from before:

    \begin{asparaenum}
        \item The edge $e_3$ has its endpoint in $e_1 \cap e_2$: We can see that this will be counted in line $6$. Whenever the outer loop gets the value $u = \source{e_1\cap e_2}$ and the innermost loop gets $v = \sink{e_3}$.
        \item The edge $e_3$ does not have its endpoint in $e_1 \cap e_2$. 
            If $\{\sink{e_1}\} = \{\sink{e_2}\} = e_1\cap e_2\cap e_3$, this case is directly checked in line $14$.
            Otherwise, this case will be counted in line $17$.\qedhere
    \end{asparaenum}
    
\end{proof}

\section{Computing open patterns} \label{sec:open}

In this section we discuss an algorithm for counting the open patterns ($21$ to $26$). The number of these patterns can be extremely large,
so an efficient algorithm must count without enumerating.

\begin{restatable}{claim}{totalcounts}\label{clm:total_counts}
\begin{equation*} 
    \sum_{i=21}^{26} \ct{i}  +  3 \sum_{i=1}^{20} \ct{i} = \sum_{e\in E(G)} \binom{\edeg{e}}{2}
\end{equation*}
\end{restatable}
\begin{proof}
    $\sum_{e\in E(G)} \binom{\edeg{e}}{2}$ gives the number of pairs of hyperedges intersecting with a hyperedge $e$. For each copy of an open pattern ($21$ to $26$) there will be one hyperedge that intersect with the other two. For each copy of a closed pattern ($1$ to $20$) each of the hyperedges intersect the other two, so it will be counted thrice.
\end{proof}

Similarly we can show the following:

\begin{restatable}{claim}{typetwofour}\label{clm:type_twofour}
\begin{equation*}
    \sum_{e \in E(G)} \binom{\Nchildren{e}}{2} = \ct{1} + \ct{4} + \ct{5} +\ct{21} +\ct{22}
\end{equation*}
\begin{equation*} 
    \sum_{e \in E(G)} \Nchildren{e}\cdot \Nintersect{e} = 2(\ct{2} + \ct{3}) + \ct{6} + \ct{7} + \ct{8} + \ct{23} +\ct{24}
\end{equation*}
\end{restatable}
\begin{proof}
    For the first equation, we are summing for every hyperedge $e$ the number of pairs that are children of $e$. Let $e_1,e_2 \in \children{e}$, we can distinguish three cases:
    \begin{asparaitem}
        \item $e_1\cap e_2 = \emptyset$: Then we either have a pattern of type $21$ or $22$.
        \item $e_1 \subset e_2$ (or vice-versa): then we have a type $1$ pattern.
        \item $e_1$ intersects $e_2$ (without containment): this leads to a type $4$ or $5$ pattern.
    \end{asparaitem}

    For the second equation, we are summing for every hyperedge $e$ the number of pairs formed by a children of $e$ and one of the hyperedges that intersect it. Let $e_1 \in \children{e}$ and $e_2$ be a hyperedge with $e_2\cap e \neq \emptyset$ that is not in $\parents{e}$ or $\children{e}$. Again we have three cases:
    
    \begin{asparaitem}
        \item $e_1\cap e_2 = \emptyset$: Then we either have a pattern of type $23$ or $24$.
        \item $e_1 \subset e_2$: This corresponds to either a pattern of type $2$ or $3$. Moreover, these patterns will be counted twice as $e_2$ intersects with $e$ and is an ancestor of $e_1$.
        \item $e_1$ intersects $e_2$ (without containment): This corresponds to patterns $6,7,8$.
    \end{asparaitem}
\end{proof}

We now show an algorithm that computes the counts of patterns $21$,$23$ and $25$. The key observation is that these patterns contain a hyperedge $e_1$ that is perfectly divided between two other hyperedges $e_2$ and $e_3$. That is, every vertex $v \in e_1$ belongs to either $e_2$ or $e_3$ (but not both). Therefore we need to count the number of such instances, differentiating the cases where $e_2$ and $e_3$ may be subsets of $e_1$. 

\begin{definition} [$\compCounts$]
    For a graph $G$ the vector of complement counts $\compCounts = (\compCounts[0],\compCounts[1],\compCounts[2])$ indicates the number of tuples $(e_1,\{e_2,e_3\})$ such that $e_1 \subseteq e_2\cup e_3$ and $e_1\cap e_2 \cap e_3 = \emptyset$. The index indicates how many out of $e_2,e_3$ are children of $e_1$.
\end{definition}

\begin{figure}[htbp]
	\centering
	\includegraphics[width=0.5\linewidth]{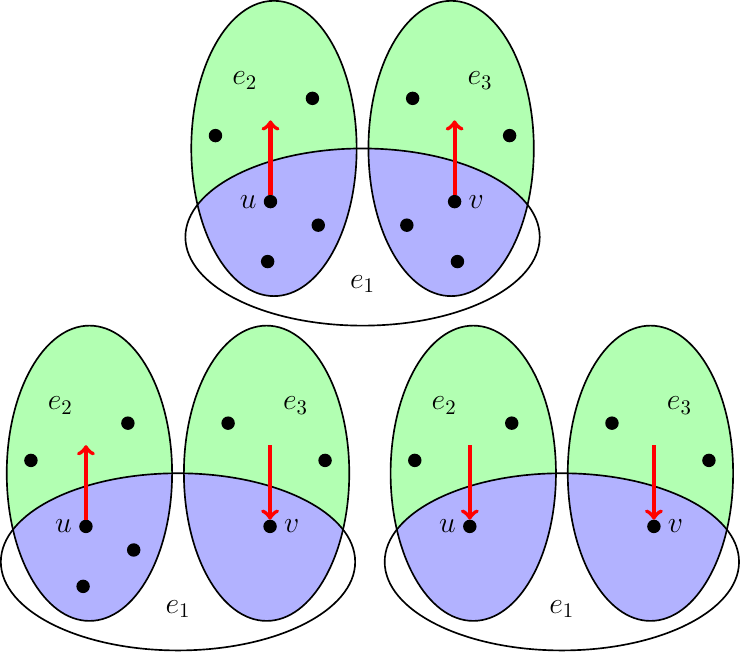}
	\caption{The three possible cases of pattern 25.
	}
	\label{fig:open_paths}
\end{figure}

We can show how to obtain the counts of patterns $21,23$ and $25$ from the complement counts:

\begin{restatable}{claim}{cc}\label{clm:cc}
\begin{equation*}
    \ct{21} = \compCounts[2]
\end{equation*}
\begin{equation*}
    \ct{23} = \compCounts[1] 
\end{equation*}
\begin{equation*}
    \ct{25} + 3\cdot \ct{17} + 2\cdot \ct{18} + \ct{19}= \compCounts[0]
\end{equation*}
\end{restatable}
\begin{proof}
    $\compCounts[2]$ corresponds with the number of patterns such that $(e_1,\{e_2,e_3\})$ such that $e_1 \subseteq e_2\cup e_3$ and $e_1\cap e_2 \cap e_3 = \emptyset$ and both $e_2,e_3 \in \children{e_1}$. This means that the only two regions that will not be empty are $e_1\cap e_2 \setminus e_3$ and $e_1\cap e_3 \setminus e_2$, which corresponds with pattern $21$.

    Similarly, $\compCounts[1]$ will have both those regions as non-empty, and either $e_2\setminus(e_1\cup e_3)$ or $e_3\setminus(e_1\cup e_2)$. Both cases corresponds with pattern $23$.

    Finally, for $\compCounts[0]$ we have that both $e_2$ and $e_3$ are not children of $e_1$. If $e_2 \cap e_3 = \emptyset$ then this necessarily leads to pattern $25$, otherwise the pattern must have the $3$ blue regions and either $0$,$1$ or $2$ of the green ones, corresponding to patterns $17,18$ and $19$ respectively. Moreover, each of the hyperedges in pattern $17$ can be the first element of the tuple, so it will be counted thrice, while in pattern $18$ there are two hyperedges that can be the first element of the tuple so it will be counted twice.
\end{proof}

We now show how to compute $\compCounts$. We can show that there are three different cases:

\begin{asparaenum}
    \item Both $e_2$ and $e_3$ are out-neighbors of a vertex in $e_1$: In this case we can start at a hyperedge $e_1$, iterate over every pair of vertices on it, and every hyperedge in the out-neighborhood of them. We can then verify if the triplet satisfy the conditions with a linear pass.
    \item Either $e_2$ or $e_3$ are not out-neighbors of a vertex in $e_1$: In this case we can look at out-neighbors of every vertex in $e_1$, and check if they contain all but one of the vertices of $e_1$ in which case we add the in-degree of that vertex. We will need to subtract the number of hyperedges starting at a vertex in $e_1$ that end in that vertex. We can precompute these amounts using Algorithm~\ref{alg:lid}.
    \item Neither $e_2$ or $e_3$ are out-neighbors of a vertex in $e_1$: This can only happen if $|e_1| = 2$, we just need to multiple the in-degrees of each vertex, subtracting the number of hyperedges that start in one and end in the other.
\end{asparaenum}

\begin{algorithm}[t!]
	\caption{\computelid$(\vec{G})$\\ Input: DAH $\vec{G} = (G,\pi)$\\Output: The array of local in-degrees \localIndegrees}\label{alg:lid}
	\begin{algorithmic}[1]
		\State $\localIndegrees \gets \{\}$
		\For{$e_1 \in E(G)$}
		\For{$u \in e_1$}
		\For{$e_2 \in \Nout{u}$}
		\If{$u = \source{e_1\cap e_2}$ \textbf{and} $\sink{e_2} \in e_1$}
		\State $v \gets \sink{e_2}$
		\State $\localIndegrees[e_1,v] \gets \localIndegrees[e_1,v] + 1$
		\EndIf
		\EndFor
		\EndFor
		\EndFor
		\State Return $\localIndegrees$
	\end{algorithmic}
\end{algorithm}

\begin{algorithm}[t!]
\caption{\openalg$(\vec{G})$ \\ Input: DAH $\vec{G} = (G,\pi)$\\Output: The vector of complement counts $\compCounts$ }\label{alg:paths}
\begin{algorithmic}[1]
\State $\compCounts \gets [0,0,0]$
\State $\localIndegrees \gets $ \computelid$(\vec{G})$
\For{$e_1 \in E(G)$}
    \For{$u,v \in e_1$}
        \For{$e_2 \in \Nout{u}$}
            \For{$e_3 \in \Nout{v}$}
                \If{$u = \source{e_1\cap e_2} $ \textbf{and} $v = \source{e_1\cap e_3} $ \textbf{and} $e_1\cap e_2\cap e_3 = \emptyset$ \textbf{and} $e_1 \subseteq e_2\cup e_3$}
                    \State $i \gets \mathbf{1}[e_1 \supset e_2] + \mathbf{1}[e_1 \supset e_3]$
                    \State $\compCounts[i] \gets \compCounts[i] + 1$ \Comment{Case $(1)$}
                \EndIf
            \EndFor
        \EndFor
    \EndFor
    \For{$u \in e_1$}
        \For{$e_2 \in \Nout{u}$}
            \If{$u = \source{e_1\cap e_2}$ and $|e_1\cap e_2| = |e_1| - 1$}
                \State $v \gets e_1\setminus e_2$
                \State $i \gets \mathbf{1}[\,e_1 \supset e_2\,]$
                    \State $\compCounts[i] \gets \compCounts[i] + \din{v}-\localIndegrees[e_1,v]$ \Comment{Case $(2)$}
            \EndIf
        \EndFor
    \EndFor
    \If{$|e_1|=2$}
        \State $(u,v) \gets e_1$
            \State $\compCounts[0] \gets \compCounts[0] + (\din{v}-\localIndegrees[e_1,v]) (\din{u}-\localIndegrees[e_1,u])$   \Comment{Case $(3)$}
    \EndIf
\EndFor
\State Return $\compCounts$
\end{algorithmic}
\end{algorithm}

Using Algorithm \ref{alg:paths} together with the previous equations we can obtain all the counts from $21$ to $26$:

\begin{restatable}[]{lemma}{open}\label{lem:open}
    There is an algorithm that, given the counts of patterns $1$ to $20$, computes the counts of patterns $21-26$ in $O(\rank^2\sum_v d_v (\dout{v})^2) = O(\inputsize \rank^2{\edegen}^2)$ time.
\end{restatable}
\begin{proof}
    The idea is to run Algorithm \ref{alg:paths} to obtain the complement counts $\compCounts$. From there, we can obtain the counts of patterns $21,23$ and $25$ using the equations of \Clm{cc}. Using the counts of pattern $21$ and $23$ we can obtain the counts for patterns $22$ and $24$ using the equations from  \Clm{type_twofour}. At this point we have the counts of all pattern $1$ to $25$, so we can use \Clm{total_counts} to obtain the counts of type $26$.

    For the runtime, we first need to analyze Algorithm \ref{alg:lid}. There are three loops, the first iterates over all hyperedges, the second over each vertex in a specific hyperedge and the third over all out-neighbors of that vertex, the check of line $5$ can be done in $O(\rank)$ time. Therefore the algorithm can be computed in time:
    \begin{align*}
        O\left(\sum_e \sum_{v\in e} \dout{v} \rank\right) = O\left(\sum_v \sum_{e \ni v} \dout{v} \rank\right) = O\left(\rank \sum_v d_v \dout{v}\right)
    \end{align*}
    Moreover we only need $O(\inputsize)$ memory for storing the dictionary.

    We just need to check Algorithm \ref{alg:paths} now. The outer loop again iterates over all hyperedges , then it splits in three cases:
    \begin{asparaitem}
        \item Case (1) (lines $4-9$): 
        The first inner loop iterates over all pairs of vertices , and the second and third over each out-neighbor of those vertices, the final check of line $7$ and the containments can be computed in time $O(\rank)$.
        \item Case (2) (lines $10-15$): Here we are only checking for one vertex and one out-neighbor. The checks can then be implemented in $O(r)$.
        \item Case (3) (lines $16-18$): This can be computed in $O(1)$.
    \end{asparaitem}
    The counts in all the equations can be computed in $O(m)$ time. Therefore the total runtime will be:
    \begin{align*}
        &\sum_e \left(O(1) + \sum_{u\in e} \dout{u} O(\rank) + \sum_{u,v \in e} \dout{u}\dout{v} O(\rank)\right) 
        \\= &O\left(\rank \sum_e \sum_{u,v \in e} \dout{u}\dout{v}\right)
        = O\left(\rank \sum_{e} \sum_{u,v \in e} (\dout{u})^2 + (\dout{v})^2\right) 
        \\= &O\left(\rank^2 \sum_e \sum_v (\dout{v})^2\right) 
        = O\left(r^2\sum_{v}\sum_{e\ni v}(\dout{v})^2\right) 
        = O\left(\rank^2 \sum_v d_v (\dout{v})^2\right)
    \end{align*}
    Only left is to prove correctness. For that we just need to show that Algorithm \ref{alg:paths} returns the correct counts for $\compCounts$. Let $(e_1,\{e_2,e_3\})$ be a tuple, we show that the algorithm will only increment the count for it if $e_1 \subseteq e_2\cup e_3$ and $e_1\cap e_2 \cap e_3 = \emptyset$. First we show that tuples not satisfying the conditions will not be counted:
    \begin{asparaitem}
        \item If $e_1 \not\subseteq e_2\cup e_3$: The last condition of the if in line $7$ prevents this case to be counted in line $9$. If $e_2$ (or $e_3$) have $|e_1 \cap e_2| = |e_1|-1$ then $e_3$ can not contain the remaining vertex, preventing case $(b)$. 
        
        Similarly, if $|e_1|=2$ is not possible for each of the hyperedges to have their endpoints in each of the vertices of $e_1$, preventing case $(c)$.
        \item If $e_1\cap e_2 \cap e_3 \neq \emptyset$: The second last condition in line $7$ prevents this case to be counted in line $9$. 
        
        If $e_2$ (or $e_3$) have $|e_1 \cap e_2| = |e_1|-1$ and $e_3$ has its endpoint in the vertex $v$ in $e_1\setminus e_2$, then it will increment $\din{v}$. However note that the vertex in  $e_1\cap e_2 \cap e_3$ must be a source in that case. And therefore this tuple will be account for in $\localIndegrees[e_1,v]$, which will cancel out in line $15$. 
        
        Similarly if $|e_1| = 2$, with $e_1=(u,v)$, if $e_2$ has its endpoint in $u$ and $e_3$ has its endpoint in $v$, then it must be the case that $e_1\cap e_2 \cap e_3 = \{u\}$ and $e_2$ will increase both $\din{v}$ and $\localIndegrees[e_1,v]$ which will cancel out on line $18$.
    \end{asparaitem}
    Only left is to verify that if a tuple follow the two properties it will be accounted for. We can look individually at the three possibles cases:
    \begin{asparaenum}
    \item Both $e_2$ and $e_3$ are out-neighbors of a vertex in $e_1$: This case will be account for in line $9$, the if in line $7$ ensures that each such tuple is only counted once. In line $8$ we check if $e_2$ or $e_3$ are children of $e_1$, in order to update the correct entry in $\compCounts$.
    
    \item Either $e_2$ or $e_3$ are not out-neighbors of a vertex in $e_1$: This can only occur if either $|e_1\cap e_3| = 1$ or $|e_1 \cap e_2| = 1$, otherwise both hyperedges will have at least one vertex in $e_1$ from which they are out-neighbors. Assume the first without loss of generality. This implies $|e_1 \setminus e_2| = 1$, and the vertex $v \in e_1 \setminus e_2$ must also be the endpoint of $e_3$, however note that this hyperedge does not contribute to $\localIndegrees[e_1,v]$ as it does not have a source in $e_1$. Therefore it will be counted on line $15$. Note that in this case only $e_2$ can be a child of $e_1$ as $|e_3|>2$ but $|e_1\cap e_3|=1$.
    
    \item Neither $e_2$ or $e_3$ are out-neighbors of a vertex in $e_1$: This can only occur when both $|e_1\cap e_2| = 1$ and $|e_1\cap e_3| = 1$, this implies $|e_1|=2$. Let $u \in e_1\cap e_2 $ and $v \in e_1\cap e_3$. Without loss of generality let $e_1=(u,v)$, note that neither $e_2$ nor $e_3$ can contribute to entries in $\localIndegrees[e_1,u]$ or $\localIndegrees[e_1,v]$. Therefore the tuple will be counted on line $18$. In this case neither of the hyperedges can be contained.\qedhere
\end{asparaenum}
\end{proof}
\end{document}